\documentclass[aps,pra,reprint,superscriptaddress]{revtex4-2}
\usepackage{tikz}
\usetikzlibrary{calc}
\usepackage{makecell}
\usepackage{array}
\usepackage[english]{babel}
\usepackage[utf8]{inputenc}
\usepackage{indentfirst}
\usepackage{amsmath}
\usepackage{amssymb}
\usepackage{amsfonts}
\usepackage{bbold}
\usepackage{amsthm}
\usepackage{physics}

\usepackage[colorlinks,urlcolor=blue,citecolor=blue,linkcolor=blue]{hyperref}

\usepackage{xcolor,graphicx}
\usepackage{dsfont}

\usepackage{subfigure}

\newcommand{\complex}{\mathbb{C}}

\newcommand{\bea}{\begin{eqnarray}}
\newcommand{\eea}{\end{eqnarray}}
\newcommand{\bel}[1]{\begin{eqnarray}\label{#1}}
\newcommand{\eel}{\end{eqnarray}}

\newtheorem{thm}{Theorem}

\newtheorem{lem}[thm]{Lemma}
\newtheorem{prop}[thm]{Proposition}
\newtheorem{cor}[thm]{Corollary}
\newtheorem{obs}[thm]{Observation}
\newtheorem{defn}{Definition}

\newcommand*{\myref}[2]{\hyperref[{#1}]{#2}} 
\newcommand{\id}{\mathds{I}} 
\newcommand{\ra}{\rangle}
\newcommand{\la}{\langle}

\newcommand{\elte}{\affiliation{MTA-ELTE ``Momentum'' Integrable Quantum Dynamics Research Group,\\ ELTE E\"otv\"os Lor\'and University \\ Pázmány P\'eter s\'et\'any 1/A, H-1117 Budapest, Hungary}}
\newcommand{\uj}{\affiliation{Institute of Theoretical Physics, Faculty of Physics, Astronomy and Applied Computer Science, Jagiellonian University \\ Ul. Łojasiewicza 11,
30-348 Krak\'ow, Poland}}
\newcommand{\ujphd}{\affiliation{Doctoral School of Exact and Natural Sciences, Jagiellonian University \\ Ul. Łojasiewicza 11, 30-348 Krak\'ow, Poland}}
\newcommand{\pancft}{\affiliation{Center for Theoretical Physics (CFT), Polish Academy of Sciences \\  Al. Lotnik{\'o}w 32/46, 02-668 Warszawa, Poland}}
\newcommand{\obuda}{\affiliation{\'Obuda University, Institute of Applied Mathematics, Bécsi út 96/B, H-1034, Budapest, Hungary
  }}

\begin{document}

\title{
  Graph restricted tensors:
building blocks for holographic networks
}
\author{Rafa{\'l} Bistro{\'n}}
\uj \ujphd
\author{M\'arton Mesty\'an}
\obuda \elte
\author{Bal\'azs Pozsgay}
\elte
\author{Karol \.Zyczkowski}
\uj \pancft

\date{\today}

\begin{abstract}
We analyze  few-body quantum states with particular correlation properties 
imposed by the requirement of maximal bipartite entanglement for selected partitions of the system into two
complementary parts. 
 A novel framework to treat this problem by encoding these constraints in a graph is advocated;
the resulting objects are called ``graph-restricted tensors''. This framework encompasses several examples previously treated in the literature, such as 1-uniform multipartite states, quantum states related 
to dual unitary operators  and absolutely maximally entangled states (AME) corresponding to 2-unitary matrices.
Original examples of presented graph-restricted tensors are motivated by tensor network models for the
 holographic principle. In concrete cases we find exact analytic solutions, demonstrating thereby that there exists
a vast landscape of non-stabilizer tensors useful for the lattice  models of holography.
\end{abstract}

\maketitle

\section{Introduction}

Tensor networks are ubiquitous in theoretical physics: they can be used to simulate quantum many-body systems in various
physical scenarios. Originally they were developed to describe ground states of local Hamiltonians with a mass-gap 
\cite{mps-intro-2020,tensor-network-review}, but these methods were later extended to describe
time-evolution as well \cite{mps-time-evolution-review}. Tensor networks found various applications in quantum
computing \cite{tensor-network-qcomp}.
A special application of tensor networks is that of the {\it holographic error correcting
  codes} \cite{holocode-review-jahn-eisert}. These tensor networks  can be used as lattice models of the
holographic principle, and also as quantum error correcting codes. 
The first example that appeared in the literature was the HaPPY code
\cite{ads-code-1}. It is a tensor network constructed from so-called perfect tensors, and it serves as a toy model of the famous
AdS/CFT correspondence \cite{Maldacena-ads-cft, Witten1998}. 

A special class of tensor networks are those which are solvable in the sense that selected physical quantities can be
computed analytically, or numerically with a small computational effort. Solvability can be established if the
individual tensors satisfy special constraints. The actual constraints depend on the physical scenario, and various
examples have already been established in the literature. 

As a first example we mention  {\it brickwork quantum circuits}: they can be understood as a special tensor network,
describing, among others, the evolution of one dimensional spin chains, 
which are solvable if the quantum gates are {\it dual unitary}   \cite{dual-unitary-review}. Analogously, there are
solvable tensor networks describing quantum states of two dimensional spin models \cite{du-peps}. Finally, the
holographic tensor networks that were studied in the literature are also solvable, because their two-point functions can
be computed \cite{holocode-review-jahn-eisert}.

In all of these examples the solvability of the tensor networks followed from special properties of the fundamental
tensors, namely that they could be interpreted as unitary or isometric linear operators for multiple arrangements of the
tensor indices \cite{sajat-multi}. Alternatively, if the tensors are viewed as multi-party quantum states, then the
solvability conditions can be translated into requirements for maximal entanglement for selected bipartitions of the
sites (tensor indices) \cite{sajat-multi}.

A special case is when there is maximal entanglement for {\it all bipartitions}; in such a case the state is called an
Absolutely Maximally Entangled (AME) state \cite{ame-review}. If one builds a tensor network from an AME states, then the
resulting network often oversimplifies while calculating interesting physical phenomena, for example two-point correlation
functions become trivial. This is the case for the dual unitary circuits  \cite{dual-unitary--bernoulli} and holographic networks \cite{ChunJun2022}. Such an observation leads to the idea of keeping the requirements of
maximal entanglement for selected bipartitions only
\cite{perfect-tangles,block-perfect-tensor,planar-OA,planar-AME,sajat-multi}. That 
way one can generate non-trivial correlation functions even in the solvable models.

In the case of holographic tensor networks, closely related alternative methods also emerged. 
Evenbly introduced the concept of hyper-invariant tensor networks (HTN), tailored to simulate AdS/CFT correspondence 
\cite{Evenbly-HyMera}. This construction uses two constituent tensors satisfying special unitarity and isometry requirements.
The resulting tensor network is expected to simulate conformal field theories at its boundary, and 
various physical quantities of
these conformal field theories (the central charges  
$c$ and scaling dimensions of primary fields $\Delta$) were also computed \cite{ChunJun2022, holocode-review-jahn-eisert}.
Later on this research field has developed into many directions \cite{Hayden:2016cfa, Qi:2018shh,
  Bao:2018pvs, Jahn2019_2, Jahn2022, Jahn2022_1,Chen:2024unp}. Recently, a construction based on dual unitary operators was advocated in
\cite{rafal-karol-mykhailo--holocode}.

The main aim of this work is to introduce a novel {representation} for tensors (quantum states), which describes the isometry/unitarity
requirements by means of graph theoretical tools. We call the resulting objects ``graph-restricted tensors'

Afterwards we
return to the lattice models of the holographic principle. We start with the requirements of exact solvability for
different regular tilings of the Poincar\'e disk. These constraints for the tensors allow for non-trivial
and hitherto previously unknown solutions. This approach leads to multiple families of tensors with tunable
parameters. Intermediate computations become even simpler than in the previous approaches \cite{rafal-karol-mykhailo--holocode},
because the proposed
examples often have smaller bond dimension.

\section{Graph-constrained tensors}
\label{sec:gct}

In this Section we present the framework of graph-constrained tensors. The framework advocated in this work covers, in
special cases, 
several examples studied earlier in the literature literature.
{The graphs that we introduce will serve as representations of the entanglement properties of particular quantum states.}

Let us first recall two standard definitions of graph theory.

\begin{defn}
  A \emph{graph} is a pair $G = (V,E)$, where $V = \{v_i\}_{i = 1}^n$ is a set of  vertices, and $E$ is a subset of $V
  \cross V$. The elements in $E$ are called edges, and we regard them as un-ordered pairs. The edge $\{v_i, v_j\}$ is
  denoted by $e_{ij}$.
\end{defn}

\begin{defn}
  A \emph{clique} of a graph $G=(V,E)$ is a subset of vertices $C \subset V$ such that for any two vertices $v_i,v_j \in C$, there exists an edge connecting them $e_{ij} \in E$. 
\end{defn}

In order to introduce graph-constrained tensors, we  also use tensor-related notations.
Let $T_{s_1,\cdots,s_n}$ denote a complex-valued tensor of order $n$ with all local dimensions equal to $d$. This tensor can be interpreted as the list of coefficients of an unnormalized quantum state $|\psi_T\ra \in ({\mathbb C}^d)^{\otimes n}$, given by
\begin{equation}
|\psi_T\ra = \sum_{i_1\cdots i_n}T_{i_1,\cdots,i_n} |i_1,\cdots,i_n\ra~.
\label{eq:psiT}
\end{equation}

Alternatively, the tensor $T$ can be interpreted as the list of matrix elements of an operator. Let us divide the labels of the indices $\{1,...,n\}$ into two arbitrary complementary subsets with 
\begin{equation}
  \begin{aligned}
    S &\equiv \{k_1,\cdots,k_m\},  \\
    \overline S &= \{1,...,n\} \setminus S \equiv \{ l_1,\cdots,l_{m'} \}\,.
  \end{aligned}
  \label{eq:bip}
\end{equation}
Then the tensor $T$ can be identified with an operator $V_{T}: (\complex^{d})^{\otimes m'} \rightarrow (\complex^{d})^{\otimes m}$,
given by
\begin{equation}
V_T = \sum_{i_1\cdots i_n}T_{i_1,\cdots,i_n} |i_{k_1},\cdots,i_{k_m}\ra\la i_{l_1}, \cdots, i_{l_{m'}}|.
\label{eq:VT}
\end{equation}
In the rest of this article, we will consider several bipartitions of indices for the same tensor. For a generic bipartition of indices \eqref{eq:bip} we introduce the concise notation $(l_1,...,l_{m'})\rightarrow (k_1,...k_m)$.

Given a tensor $T$ and a bipartition of its indices $(l_1,...,l_{m'}) \rightarrow (k_1,...,k_m) $, we define the
reduction of $T$ with respect to this bipartition as
\begin{equation}
  \rho^{j_{k_1},...,j_{k_m}}_{i_{k_1},...,i_{k_{m}}}\!\! \equiv \!\!\! \sum_{\substack{i_{l_1},\cdots,i_{l_{m'}} \\ j_{l_1},\dots, j_{l_{m'}} }} \!\!\!\!\! \delta_{i_{l_1}}^{j_{l_1}}\cdots \delta_{i_{l_{m'}}}^{j_{l_{m'}}}  T_{j_1,\cdots,j_n} (T_{i_1,\cdots, i_n})^*  \,.
\end{equation}
Similarly to the two ways of interpreting $T$, there are two ways to interpret $\rho$ as well. On one hand, it can be considered as the matrix of the reduced density operator $\text{Tr}_{l_{1},...,l_{m'}} |\psi_T\ra \la \psi_T|$ corresponding to the state \eqref{eq:psiT}. On the other hand, it can be considered as the matrix of the operator $V_T V_T^{\dagger}$, with $V_{T}$ given in \eqref{eq:VT}. 

Now we are prepared to define graph-constrained tensors.

\begin{defn}\label{def:graph_res}
  Let $G=(V,E)$ be a graph with $n$ vertices $V=\{v_{1},...,v_{n}\}$. A tensor $T_{i_1,\cdots,i_n}$  or order $n$ is constrained by the graph $G$ if for any bipartition $(l_{1},...,l_{m'}) \rightarrow (k_1,...,k_{m})$ of the tensor indices with $m \le n/2$, for which the set of vertices $\{v_{k_1},\cdots,v_{k_m}\}$ is a clique of $G$, the reduction with respect to the bipartition is proportional to identity, i.e.,
\begin{equation}
\label{eq_start}
\rho_{i_{k_1},...,i_{k_m}}^{j_{k_1},...,j_{k_m}} \propto \delta_{i_{k_1}}^{j_{k_1}}  \cdots \delta_{i_{k_m}}^{j_{k_m}}\,. 
\end{equation}
Furthermore, if for any bipartition of $T$ satisfying \eqref{eq_start}, there exists a corresponding clique within graph $G$, $T$ is said to be faithfully constrained by graph $G$.
\end{defn}

In the quantum state interpretation of $T$, the equation \eqref{eq_start} states that the reduced density matrix  is maximally mixed $\Tr_{i_{l_1},...,l_{m'}} |\psi_T\ra\la\psi_T| \propto \id$. In the operator interpretation \eqref{eq:VT}, the condition \eqref{eq_start} is equivalent to $V_T V_T^{\dagger} \propto \id$, which means that $V_T$ is proportional to an isometry.

We note that the above definition is compatible with the fact that any subset
of a clique is also a clique. If the reduction of a tensor is proportional to identity, then it will remain proportional to identity after the contraction of further indices.

Having two graphs constrained tensors we may sometimes derive modest constraints for their contractions.

\begin{prop}
Let $T^{(1)}$ and $T^{(2)}$ be two tensors  constrained by graphs $G^{(1)}$, $G^{(2)}$. If $T^{(1)}$ and $T^{(2)}$ are contracted on some indices corresponding to one clique in each graph, $C^{(1)}, C^{(2)}$, then the resulting tensor is constrained by graph $G$ which is a disjoint union of $G^{(1)}$ and $G^{(2)}$ both with removed vertices corresponding to contracted indices. 
\end{prop}
\begin{proof}
Let us contract the tensors $T^{(1)}$ and $T^{(2)}$ corresponding to some pairing of the indices in the cliques
$C^{(1)}$ and $C^{(2)}$. Let us furthermore select an additional clique in $G^{(1)}$, which is assumed to be have zero
overlap with $C^{(1)}$; we denote this additional clique by $D^{(1)}$. We could also select the additional clique in
$G^{(2)}$, in which case one has to repeat the argument below.

Let us now consider the reduction of the contracted
tensor with respect to $D^{(1)}$. As we compute the reduction, we may choose to contract the remaining indices in
$T^{(2)}$ and its complex conjugate. Due to the fact that $C^{(2)}$ is a clique, this first and partial step of the
reduction process yields an identity matrix for the indices belonging to $C^{(2)}$ in $T^{(2)}$ and its complex
conjugate. 
holds of course for both the tensors and their complex conjugates.
Having found the identity matrix in the previous step, we now recall that $T^{(1)}$ and $T^{(2)}$ were contracted on these indices,
so the identity obtained from $C^{(2)}$ connects the indices belonging to $C^{(1)}$ in
both $T^{(1)}$ and its complex conjugate. Finally we perform the reduction over those indices which do not belong to
either $C^{(1)}$ or $D^{(1)}$. Altogether we get the reduction over all indices which are not in $D^{(1)}$. Using the
fact that $D^{(1)}$ was originally a clique in $G^{(1)}$, this partial reduction also yields the identity matrix. With
this we proved that the partial reduction to $D^{(1)}$ yields maximal entanglement also in the contracted tensor.
\end{proof}

The notion of graph-constrained tensor is not to be confused with graph states \cite{AME-graph}. {Also, the edges of the
graph should not be confused with Bell pairs, which are also often denoted by connecting two points in similar
depictions of quantum states.}

In Table \ref{tab1} we present a few examples of well-known structures that can be described by this framework, and below we
discuss them in more details. Every example is a 4 index tensor, and the dimension of the individual spaces is not specified.

\begin{table*}[t]
  \begin{tabular}{|c||c|c|c|c|}
  \hline
Graph &
        \raisebox{-.5\height}{
        \begin{tikzpicture}[scale=0.5, every node/.style={font=\small}]
  \foreach \i in {1,...,4}{
   \coordinate (v\i) at ({225+(\i-1)*90}:2cm);
  }
   \foreach \i in {1,...,4}{
     \node[draw,fill=blue, circle, inner sep=2pt] at (v\i) {};
      \node at ($(v\i)+ ({225+(\i-1)*90}:0.6cm)$) {$\i$};
  }   
\end{tikzpicture}}
    &
      \raisebox{-.5\height}{
      \begin{tikzpicture}[scale=0.5, every node/.style={font=\small}]
  \foreach \i in {1,...,4}{
   \coordinate (v\i) at ({225+(\i-1)*90}:2cm);
 }
   \draw[thick] (v1) -- (v2);        \draw[thick] (v4) -- (v3);   
   \foreach \i in {1,...,4}{
     \node[draw,fill=blue, circle, inner sep=2pt] at (v\i) {};
    \node at ($(v\i)+ ({225+(\i-1)*90}:0.6cm)$) {$\i$};
  }   
\end{tikzpicture}    }
    &
      \raisebox{-.5\height}{
      \begin{tikzpicture}[scale=0.5, every node/.style={font=\small}]
  \foreach \i in {1,...,4}{
   \coordinate (v\i) at ({225+(\i-1)*90}:2cm);
 }
   \draw[thick] (v1) -- (v4);        \draw[thick] (v2) -- (v3);
       \draw[thick] (v1) -- (v2);        \draw[thick] (v4) -- (v3);   
   \foreach \i in {1,...,4}{
     \node[draw,fill=blue, circle, inner sep=2pt] at (v\i) {};
     \node at ($(v\i)+ ({225+(\i-1)*90}:0.6cm)$) {$\i$};
  }   
\end{tikzpicture}    }
    &
      \raisebox{-.5\height}{
      \begin{tikzpicture}[scale=0.5, every node/.style={font=\small}]
  \foreach \i in {1,...,4}{
   \coordinate (v\i) at ({225+(\i-1)*90}:2cm);
 }
   \draw[thick] (v1) -- (v4);        \draw[thick] (v2) -- (v3);
   \draw[thick] (v1) -- (v2);        \draw[thick] (v4) -- (v3);
   \draw[thick] (v1) -- (v3);        \draw[thick] (v2) -- (v4);
      \foreach \i in {1,...,4}{
     \node[draw,fill=blue, circle, inner sep=2pt] at (v\i) {};
     \node at ($(v\i)+ ({225+(\i-1)*90}:0.6cm)$) {$\i$};
   }
  \end{tikzpicture} }
      \\
 \hline
 State & 1 uniform & \makecell{maximal entanglement \\ for bipartition \\ $(1,2)\leftrightarrow  (3,4)$} &
       \makecell{maximal entanglement \\ for bipartitions\\ $(1,2)\leftrightarrow (3,4)$ and\\  $(1,4) \leftrightarrow (2,3)$}
        & \makecell{maximally entangled \\ for all bipartitions  \\ (AME)}\\
     \hline
 \makecell{Operator\\  interpretation}  & \makecell{$1\to(2,3,4)$ isometry}
        & unitary & dual-unitary & \makecell{two-unitary \\ (perfect tensor)} \\
 \hline
\end{tabular}
\caption{Graphs with four vertices, encoding graph constrained tensors. These correspond to 
  4-partite quantum states $|\psi_T\rangle \in H_d^{\otimes 4}$, and two-site operators, or bipartite matrices of order
  $d^2 \times d^2$. The cliques of the graphs are those subsets which are maximally entangled with the
  complement. In the first case only single sites form cliques, and we obtain 1-uniform states. In the second states
  there are two cliques: the pairs $(1,2)$ and $(3,4)$; accordingly we get maximal entanglement for the bipartition
  $(1,2)\cup(3,4)$, and a unitary operator acting as $(1,2)\to (3,4)$. In the third example we obtain the dual unitary
  operators, and in the last example of a full-graph -- the AME states or perfect tensors.} 
\label{tab1}
\end{table*}

The first non-trivial example is a $1$ uniform tensor, which is encoded by an empty graph. In this graph the only
cliques are those of the isolated sites, therefore the resulting conditions mean that
every local party is maximally entangled with the rest of the system.
A well known example for such a state is the GHZ state. In the case of tensor of order 4, the resulting isometry
conditions take the form
\begin{equation}
  \sum_{i_2 i_3 i_4} T_{i_1,i_2,i_3,i_4} (T_{i_1,i_2,i_3,i_4})^* =
\delta_{i_1, i_1'},
\end{equation}
together with 3 analogous equations, corresponding to the sites $2, 3, 4$.

The second example describes a state which has maximal entanglement for a single bipartition into two equal halves, in
this case $(1,2)\leftrightarrow (3,4)$. Interpreted as an operator acting from subsystem $(1,2)$ to subsystem $(3,4)$ we
obtain a unitary matrix. Therefore, the single condition that we get is
\begin{equation}
  \sum_{i_3,i_4} T_{i_1,i_2,i_3,i_4}
(T_{i_1',i_2',i_3, i_4})^* = \delta_{i_1,i_1'}\delta_{i_2,i_2'}
\end{equation}

The next example is a planar two-uniform tensor, also called
 block-perfect tensor
\cite{perfect-tangles,block-perfect-tensor,planar-AME}. The graph is given by a square, and the largest cliques are
given by the pairs of neighboring sites. Correspondingly, this state has maximal entanglement for two different
bipartitions, which cut the square into two halves parallel to its sides. Viewed as an operator these tensors are
dual-unitary \cite{dual-unitary-review}.

Our final example is the perfect tensor, which corresponds to an absolutely maximally entangled state \cite{AME-1}.
As an operator it was also called a two-unitary matrix \cite{AMEcomb2}, and it serves as an isometry for any grouping
of indices. The constraints are represented by the complete graph.

We note that the framework of graph-constrained tensors can not describe all possible situations of physical
relevance. A complete framework can be given by the so-called hypergraph-constrained tensor, which we introduce in
Appendix \ref{app:hyper}. Hypergraph-constrained tensors are a natural extension of the concept of graph-constrained
tensors. However, in most situations of physical interest the simpler representation by a single graph is sufficient.
 
 {In the rest of this work, we focus on two new types of graph-constrained tensors their minimal construction.
Next two section are devoted to study minimal constructions with two-dimensional local spaces.
A family of solutions with higher-dimensional local dimensions is presented in Appendix \ref{sec:n-gon}.}

\section{Planar pentagonal tensor}\label{sec:penta}

{In this Section we study a specific class of graph-constrained tensors.  Before going to the
technical details, we can summarize the problem in the following way: We intend to find all quantum states of 5 qubits,
such that the states have the geometrical symmetries of the pentagon, and each pair of neighboring qubits is maximally
entangled with the remaining 3 qubits. This problem is then formalized as follows.}

\begin{figure}[h]
  \centering
  \begin{tikzpicture}[scale=0.5, every node/.style={font=\small}]
  \foreach \i in {1,...,5}{
    \coordinate (v\i) at ({90-(\i-1)*72}:2cm);
  }
   \foreach \i/\j in {1/2,2/3,3/4,4/5,5/1}{
     \draw[thick] (v\i) -- (v\j);
     
   \foreach \i in {1,...,5}{
    \node[draw,fill=blue, circle, inner sep=2pt] at (v\i) {};
  }   
  }
\end{tikzpicture}
\caption{Pentagonal graph encoding constraints of planar 2-uniform tensor or order $5$ \cite{planar_k_uniform_states}.}
\label{fig:penta}
\end{figure}
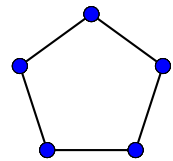

We are looking for tensors that are of order five, and we assume
that each subspace has dimension 2. The constraints are encoded by the pentagonal graph, see
Fig. \ref{fig:penta} and describe a planar 2-uniform tensor \cite{planar_k_uniform_states}.
In this graph the pairs of neighbors are the largest cliques, and the requirement implies that such tensors are
specific variants of ``planar perfect tensors'' which were called various names in previous literature \cite{perfect-tangles,block-perfect-tensor,planar-AME}.

We intend to respect the geometrical symmetries of the pentagon. Therefore, we are looking for a tensor $T$ of order 5 with
components  $T_{s_1,s_2,s_3,s_4,s_5}$  such that the tensor itself is invariant under the symmetry group of the
pentagon. Due the geometrical symmetries it is sufficient to formulate one constraint of maximal entanglement, which reads
\begin{equation}
  \label{rho12}
  (\rho_{12})_{s_1,s_2}^{s_1',s_{2}'} = \delta_{s_1}^{s_{1}'} \delta_{s_2}^{s_2'},
\end{equation}
where $\rho_{12}$ corresponds to the bipartition $(3,4,5)\rightarrow(1,2)$. At the same time, we do not require that
$\rho_{13}$ is proportional to the identity matrix. Therefore, we allow that the subsystem corresponding to sites $(1,3)$ is
not maximally entangled with its complement.

A rank 5 qubit tensor has a total number of $2^5=32$ components. However, the geometrical symmetries imply that there are only
8 independent ones. In the case of qubits the rotational symmetry is enough to narrow down the list of independent
components. They can be chosen as
\begin{equation}
  \begin{split}
    &    T_{00000}, T_{11111} ,    T_{00001}, T_{01111} ,  \\
    &T_{00011}, T_{00111} ,    T_{00101}, T_{01011}
  \end{split}
\end{equation}
One can directly check that every component not present in the above list can be identified with one of these elements,
by using rotational symmetry on the indices.

We computed the reduced density matrix $\rho_{12}$ as an expression of these 8 independent components, and enforced
condition \eqref{rho12}. For simplicity we restricted our attention to tensors with purely real components. 

By using program \texttt{Mathematica} \cite{MathematicaTensorsPaper} it was possible to find real solutions of eq. \eqref{rho12}.
Writing out the components we end up with $7$ equations for $8$ tensor coefficients. 

 By using reparametrisations we
determined that the algebraic variety of found solutions consists of two disjunct components. One of the components is a
two-parameter family of solutions, whereas the other component is a one-parameter family. 

Once a certain solution is found, one can still apply global $SU(2)$ transformations to obtain new
solutions. Since we restrict ourselves to tensors with real components, the only allowed $SU(2)$ rotations are those
where the representant is also purely real, and they are given by
\begin{equation}
\otimes_{j=1}^5 S\quad \text{with} \quad   S=\exp(i\phi Y),
\end{equation}
where $Y$ is the standard Pauli matrix corresponding to the $y$ direction. 

Using the symmetry transformation one can simplify the solutions even further. One of the components is found to be a one-parameter
family of tensors, and the other component is an isolated point.

We also computed the reduced density matrix $\rho_{13}$ for the solutions, which was not required to be an the identity matrix. However, specifically for the one-parameter family of solutions it was the case.
This means that the state is actually an $AME(5,2)$ with a free
parameter. We confirmed that the solution coincides with the one-parameter family of AME states found recently in
\cite{arul-AME5d} (see also \cite{tan-AME52}). {For completeness, we present our solution described by co-efficients}
\begin{equation}
  \begin{split}
    T_{00000}=T_{00101}=-T_{00011}=-T_{01111}&=\sin(\theta)\\
    T_{11111}=T_{01011}=-T_{00111}=-T_{00001}&=\cos(\theta),
  \end{split}
\end{equation}
where $\theta$ is a real parameter. 

In the case of the isolated point we also computed $\rho_{13}$ which differs from the identity matrix, implying that
this tensor is imperfect. Its components read
\begin{equation}
  \begin{aligned}
    T_{00000}&= T_{11111} + \frac 32 =    \frac{1}{4} \left(\sqrt{10 \sqrt{5}-22}+3\right)\,,\\
    T_{00001}&= -T_{01111} = -  \frac{1}{4} B\,,\\
    T_{00011}&= T_{00111} - \frac{B}{2} =  \frac{1}{4} \left(-B+2    \sqrt{B}\right)\,,\\
    T_{00101}&= T_{01011} + \frac{1}{2} =  \frac{1}{4} \left(1-\sqrt{2( \sqrt{5}-1)}\right)\,,\\
  \end{aligned}
  \label{eq:imp5}
\end{equation}
where $B\equiv\sqrt{5}-2$. The remaining coefficients are determined by rotational invariance.

The reductions have elegant form
\begin{equation}
  \rho_{12} = \mathbb 1_2\,,\quad \rho_{13} =
  \begin{pmatrix}
    \alpha & 0 & 0 & \beta \\
    0 & \gamma & \delta & 0 \\
    0 & \delta & \gamma & 0 \\
    \beta & 0 & 0 & \alpha
  \end{pmatrix}\,,
\end{equation}
with
\begin{equation}
  \begin{aligned}
  \alpha & = \! \frac{\sqrt{5} + 3}{4}\,,\quad
  \beta  =  \frac{\sqrt{5} - 1}{4}\,, \\
  \gamma  & =  \frac{5 - \sqrt{5}}{4}\,,\quad
  \delta  =  \frac{1 - \sqrt{5}}{4}\,.
\end{aligned}
\end{equation}

{Finally, in Appendix \ref{sec:n-gon} we present a general construction of such tensors with local dimension $d^2$,
where $d\geq 2$. Furtheremore, these tensors need not be pentagonal but with arbitrary number of local subsystems, $n \geq 4$.
}

\section{Planar hexagonal tensors}

\label{sec:hexagonal}

In this Section we treat a new family of graph constrained tensors. The constraints are motivated by application to the
holographic tensor networks described in the following Section.

The tensors in question have 7 indices, and the constraints are encoded by the graph from Fig. \ref{fig:hexa}. Here the 7 vertices
are arranged such that 6 of them form the corners of a regular hexagon, and one additional vertex is placed to the
center. The largest cliques in the graph correspond to triangles: three-element subsets consisting of two neighboring
corners and the central vertex. 

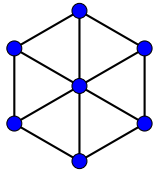
\begin{figure}[h]
  \centering
  \begin{tikzpicture}[scale=0.5, every node/.style={font=\small}]
  \foreach \i in {1,...,6}{
    \coordinate (v\i) at ({90-(\i-1)*60}:2cm);
  }
  \coordinate (v7) at (0:0cm);
  \foreach \i/\j in {1/2,2/3,3/4,4/5,5/6,6/1}{
    \draw[thick] (v\i) -- (v\j);
    \draw[thick] (v\i) -- (0:0cm);
  }
  \foreach \i in {1,...,7}{
    \node[draw,fill=blue, circle, inner sep=2pt] at (v\i) {};
  }  
\end{tikzpicture}
\caption{A graph with 7 vertices encoding the constraints for the tensors \eqref{eq:tensor}. The central node
  corresponds to the index $s_0$ from \eqref{eq:tensor} while other nodes to the remaining indices.}
\label{fig:hexa}
\end{figure}

The coefficients of the tensors are denoted by
\begin{equation}
  T^{s_0}_{s_1, s_2, s_3, s_4, s_5, s_6}\,,
  \label{eq:tensor}
\end{equation}
where the index $s_0$ stands for the site at the middle, and the indices $s_j$ with $j=1,\dots,6$ correspond to the
sites of the hexagon. We call $s_0$ the bulk index and $s_j$ the bond indices; the reason for this will become clear in
the next Section.
For simplicity we will consider only two dimensional local spaces, therefore $s_j$ can be either $0$ or $1$.

In the following we will
first list the necessary properties of the tensors and then provide an overview of the solutions we have found.

\subsection{Necessary properties of the tensors}

We require that the tensors are rotationally symmetric, in the sense of the geometric arrangement of the graph
above. This means that
\begin{equation}
  T^{s_0}_{s_1, s_2, s_3, s_4, s_5, s_6} = T^{s_0}_{s_2, s_3, s_4,s_5,s_6,s_1}\,.
  \label{eq:rotinv}
\end{equation}
Once this symmetry property is enforced, it is enough to consider one constraint for the maximal entanglement. For
example, we can choose the bi-partition $(0,1,2)\cup(3,4,5,6)$. Correspondingly, we define the reduction
\begin{equation}
  (\rho_{012})_{s_0s_1s_2}^{s_0's_1's_2'}\!\!=\!\! \sum_{\{p_j\}}\! T^{s_0}_{s_1, s_2, p_3, p_4, p_5, p_6} (T^{s_{0}'}_{s_1', s_2',p_3, p_4, p_5, p_6})^*\,,
  \label{eq:rho012}
\end{equation}
where the summation goes from $0$ to $1$ for each index $p_3,p_4,p_5,p_6$. 

Then the isometry condition with respect to this bipartition reads
\begin{equation}
 (\rho_{012})_{s_0,s_1,s_2}^{s_0',s_1',s_2'}=  \delta_{s_0}^{s_0'} \delta_{s_1}^{s_1'} \delta_{s_2}^{s_2'}.
  \label{eq:isom}
\end{equation}

Having established the constraint, we proceed to discuss the entanglement for other possible bi-partitions. In
the case of 7 qubits there is no perfect tensor \cite{AME-7qbits-nogo}, therefore we will always find a bi-partition
with not maximal entanglement. To characterize entanglement on various bi-partitions we utilize R\'enyi-2 entropy. More specifically we will compute
\begin{equation}
  \Delta s_{ijk} \equiv \Tr \rho_{0ij}^2 - 1/8 \,,
  \label{eq:renyi}
\end{equation}
where the reduced density matrices $\rho_{ijk}$ are defined analogously to $\rho_{012}$ \eqref{eq:rho012} but with
$s_i$, $s_j$ and $s_k$ being retained instead of $s_0$, $s_1$ and $s_2$. Due to rotational invariance it is sufficient
to compute $\Delta s_{013}$, $\Delta s_{014}$, $\Delta s_{123}$, $\Delta s_{124}$ and $\Delta s_{135}$.

\subsection{Solutions to the isometry condition}

Let us focus on tensors with local dimension equal $2$ and real coefficients.
Furthermore, we look for solutions to the isometry condition
\eqref{eq:isom} that have rotational invariance in the bond 
indices \eqref{eq:rotinv}. Under these conditions the number of different coefficients
$T^{s_0}_{s_1,s_2,s_3,s_4,s_5,s_6}$ is 28 and the number of independent equations is 33. 

Such a system is too complicated to be solved analytically in general. Thus further restrictions are imposed. We require the
solutions to be invariant under simultaneously flipping all indices: 
\begin{equation}
  T^{s_0}_{s_1,s_2,s_3,s_4,s_5,s_6}  = T^{\bar s_0}_{\bar s_1, \bar s_2, \bar s_3, \bar s_4, \bar s_5, \bar s_6}\,,
  \label{eq:z2inv}
\end{equation}
where $\bar s = 1 - s$. We also require that the tensors should be invariant under spatial reflection over the main diagonal ,
\begin{equation}
  T^{s_0}_{s_1,s_2,s_3,s_4,s_5,s_6} = T^{s_0}_{s_{1},s_{6},s_{5},s_{4},s_{3},s_{2}}\,.
  \label{eq:reflinv}
\end{equation}
With these two extra conditions the number of unknown coefficients is reduced to $13$.
The independent components can be chosen as follows:
\begin{equation}
  \label{aT}
  \begin{aligned}
    a_{1} &\equiv T^{0}_{000000},\quad     a_{2}\equiv T^{1}_{000000},\quad
    a_{3} \equiv T^{0}_{000001}\\
    a_{4} &\equiv T^{1}_{000001},\quad
    a_{5} \equiv T^{0}_{000011} ,\quad
    a_{6} \equiv T^{1}_{000011} \\
    a_{7} &\equiv T^{0}_{000101} ,\quad
    a_{8} \equiv T^{1}_{000101} ,\quad
    a_{9} \equiv T^{0}_{001001}\\
    a_{10} &\equiv T^{1}_{001001} ,\quad
             a_{11} \equiv T^{0}_{000111},\quad
              a_{13} \equiv T^{0}_{010101} \\
    a_{12} &\equiv T^{0}_{001011} = T^{0}_{001101}\\
  \end{aligned}
\end{equation}
All other tensor components are obtained by either rotational and/or spin reflection invariance.
We found that the number of independent
equations is reduced to $14$.

We used a combination of numerical and analytical methods to find the solutions to the conditions \eqref{eq:isom}, with
the symmetry requirements imposed. The numerical method consisted of starting with a randomly chosen tensor, and
numerically finding the minimum of a cost function $\Delta s_{012}$ from \eqref{eq:renyi}.
Afterwards, having collected numerical examples for the
solutions, we made an attempt to observe patterns between the coefficients of the tensors. In various cases we found
certain relations, which subsequently allowed us to find explicit solutions. We used the programs \texttt{octave} and
\texttt{Mathematica} \cite{MathematicaTensorsPaper}.

In order to get an overall picture of the solutions, we plot the R\'enyi entropies of the corresponding reduced density matrices, $\Delta s_{013}$ and $\Delta s_{014}$ \eqref{eq:renyi} in Figure \ref{fig:entr}. The three emerging lines in the plot correspond to one-parameter classes of solutions.
\begin{itemize}
\item Type I solutions are the most interesting for our purposes. These solutions are truly imperfect: they do not have
  maximal entanglement for those bipartitions, for which we do not impose constraints.
 Exact algebraic expressions found for this family of solutions are presented below.
  \item Type II solutions:  we could not find analytical expressions corresponding to this type. 
  \item Type III solutions are isometries also with respect to the bipartition $(0,1,3)\rightarrow (2,4,5,6)$, which is
    signified by $\Delta s_{0,1,3}=0$.
     We found analytic solutions for this type, shown in Appendix \ref{app:solutions}.
\end{itemize}

In addition to the three one-parameter families Fig. \ref{fig:entr}  shows also two isolated points, P1 an P2. Both
points correspond to solutions that only satisfy the required isometry. An explicit expression corresponding to point P2,
 is shown in Appendix \ref{app:solutions}.

\subsection{Type I. solutions}

These solutions of type I for the tensor $T$ are written in the variables  introduced in \eqref{aT} with additional relations
\begin{equation}
  \begin{aligned}
   a_2 - a_1  &= a_5 - a_6 = a_7 - a_8 = \frac{(-1)^j}{8\sqrt{2}}\,, \\
  \label{}
   a_4 - a_3  &=  \frac{(-1)^k}{4\sqrt{2}}\,, 
 \end{aligned}
 \label{eq:typei}
\end{equation}
where $j$ and $k$ are either $0$ or $1$. \verb!Mathematica! can find several analytic solutions of this type. We give a one-parameter example here, with $j=k=0$:
\begin{equation}
  \label{ak}
  \begin{aligned}
a_1 &= \frac 18 \left(-\sqrt 2 + 24 a  \mp 4D \right) \\
a_3 &= \frac 1{16} \left(-\sqrt 2 \mp 4D \right) \\
a_5 &= \frac 1{16} \left( \sqrt{2} - 16 a \mp 8D  \right)\\
a_7 &= a \\
a_9 &= \frac 1{16} \left(- \sqrt{2} - 16 a \pm 8D  \right) \\
a_{10} &= \frac 1{8} \left(\sqrt{2} - 8 a \pm 4D  \right) \\
a_{11} &= \frac 1{16} \left(\sqrt{2} - 32 a \pm 4D  \right) \\
a_{12} &= \frac 1{16} \left(-\sqrt{2} + 32 a \pm 4D  \right) \\
a_{13} &= \frac 1{16} \left(-\sqrt{2} + 32 a \mp 12D  \right) 
	\end{aligned}
\end{equation}
where $a$ is a real parameter satisfying $0 < a < \frac{\sqrt{2}}{16}$
and $D=\sqrt{a (\sqrt{2} - 16a)}$.
The entropies corresponding to these tensors, marked in Fig. \ref{fig:entr}, constituting a family of $7$-qubit states, read
\begin{equation}
  \begin{aligned}
  \Delta s_{013} &= \frac{1}{64} + a  \left( \sqrt{2} - 16 a \left( 5 + 64 a \left( 8 a - \sqrt{2} \right) \right) \right)\,, \\
  \Delta s_{014} &= \frac{3}{32} = 0.9375\,.
\end{aligned}
\end{equation}

\begin{figure}[t]
  \centering
  \includegraphics[width=85mm]{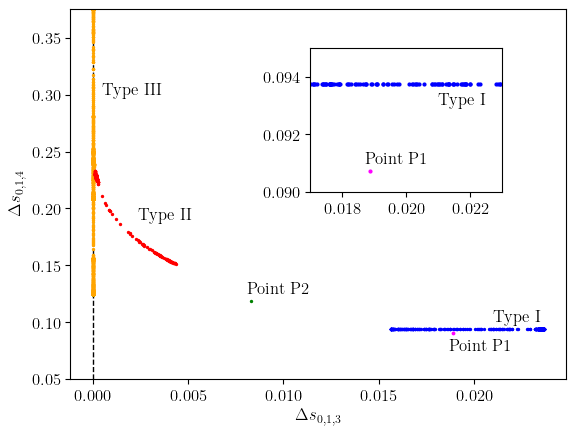}
  \caption{R\'enyi-2 entropies \eqref{eq:renyi} of the reduced density matrices $\rho_{013}$ and $\rho_{014}$
    corresponding to $7$-index tensors $T_{s_1,s_2,s_3,s_4,s_5,s_6}^{s_0}$ that satisfy the isometry condition
    \eqref{eq:isom} as well as rotational, spin-flip and spatial reflection invariance.
Every point denotes a numerical solution to the imposed constraints.
    The colors blue, red and orange corresponds to tensors belonging to first, second and third type receptively,
      except of two isolated points P1 and P2.  Inset: magnification of the same plot around the isolated point P1.}
  \label{fig:entr}
\end{figure}

{Furtheremore, in Appendix \ref{sec:n-gon} we present a general construction of such tensors with local dimension bulk subsystem equal $d_1$ and of boundary subsystems equal $d_1 d_2^2$,
where $d_1\geq 2$ and $d_2 \geq 2$. Moreover these tensors need not be hexagonal but with arbitrary number local boundary subsystems $n \geq 4$.
}

\section{Application in the AdS/CFT correspondence}

In this Section, we show how to build holographic tensor networks using the tensors introduced in previous two
Sections. Since the presented results are a generalizations of the arguments from \cite{rafal-karol-mykhailo--holocode}, we refer the reader therein for detailed discussion of
underlying concepts and employed arguments.

In the bulk of the discussion, for the clarity of the presentation, we focus on the $\{6,k\}$ tilings of the Poincar\'e
disk.This means that we treat a regular tessellation with hexagons, $k \geq 4$ of which meet at each vertex. However,
the presented arguments also hold for pentagons and polygons with a larger number of edges.   
In our construction, the above discussed tensors are sometimes considered as isometries mapping subsystems as
$(0)\to(1,2,3,\cdots)$, $(0,1)\to(2,3,\cdots)$ or $(0,1,2)\to(3,\cdots)$. In each of these cases different rescaling of
the tensor is necessary. 
For the sake of clarity, we omit these rescalings to restore them only when necessary. 

We consider a tensor network constructed in a vertex inflation manner \cite{Central_charges_and_scaling}. In the first step, we place one tensor $T^{s_0}_{s_1,s_2,s_3,s_4,s_5,s_6}$ in a tile, identifying $s_0$ index with the bulk degree of freedom. The remaining indices are identified with boundary degrees of freedom, each index corresponding to one edge of the tile. 
In the consecutive steps to enlarge the network, we place the same tensor on all tiles with a common edge or vertex to the ones already occupied by a tensor network, one by one. 
In each sub-step, we contract indices of neighboring tensors, which correspond to a common edge. Since in each new tensor we contract at most two indices with the tensor network, which happens if the tensor "closes the loop" around a tessellation vertex, we can interpret new tensor as isometry $(0,1)\to(2,3,4,5,6)$ or $(0,1,2)\to(3,4,5,6)$ from new bulk degree of freedom ($0$) and old boundary degrees of freedom ($1$ or $1,2$) into new boundary degrees of freedom. Thus, the entire tensor network $\mathcal{T}$, constructed in multiple rounds of such expansions, serves as an isometry from the bulk to the boundary.

In our work, we are primarily interested in correlation functions induced by the image of some bulk operator $\mathcal{O}$ on chosen points in the boundary.
To calculate them, we first map the bra and ket indices of the operator by the tensor network and its conjugate to obtain the boundary operator, and then calculate the expectation value of some probing observables $\{v_n\}$ corresponding to the chosen boundary subsystem. Thus, the general formula for the correlation function reads
\begin{equation}
\label{eq:gen_correlations}
\la \phi(x_1) \cdots \phi(x_n) \ra := \text{Tr}\left[ \mathcal{T}\mathcal{O}\mathcal{T}^{\dagger} (v_1 \otimes \cdots \otimes v_n \otimes \id) \right],
\end{equation}
where the identity is placed on the remaining boundary subsystems.

\begin{figure}[t]
  \centering
      \includegraphics{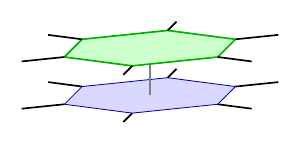}
  \caption{Contraction of bulk indices between planar hexagonal tensor (violet) and its conjugate (green), which appears while calculating  correlation functions for bulk subsystems unoccupied by bulk operator \eqref{eq:gen_correlations}}
  \label{fig:hexa1}
\end{figure}

Furthermore, we set the observables $v_i$ to be traceless. The reason for that is twofold. First, one can check by repeatedly applying the isometry conditions of tensors $T$ that in the calculation of the one-point correlation function the entire tensor network is reduced, and the result is given by 
\begin{equation*}
\la \phi(x_1) \ra = \text{Tr}[v_1]\text{Tr}[\mathcal{O}],
\end{equation*}
for any bulk operator $O$, while in conformal field theories one-point correlation function is typically equal to zero.
 Furthermore, higher point correlation functions may factorize into products of $\text{Tr}[v_i]$ terms, without any relation to distances between observables, which is unphysical as well. 

To describe second and higher-order correlation functions, it is useful to define a notion of a \textit{path} between
two distinct boundary or bulk indices.
A path is a set of tensors such that each path's tensor is connected with two other path's tensors on non-adjacent edges, or is connected with only one path's tensors and possess a distinct index corresponding to the begging or the end of the path. 
Loosely 
speaking a path is a connection built from tensors, which does not take sharp turns on the tiles' vertices, as presented in Fig. 6 in \cite{rafal-karol-mykhailo--holocode}.

A special type of path is a \textit{geodesic path}, where for each path's tensor, one edge of the tile belongs to the same geodesic on Poincar\'e disk, or there exists a geodesic which goes through the tiles.  Loosely speaking geodesic path takes the same turn at each node.
Equipped with these notions, we are ready to state the following
\begin{lem}
Consider a regular tiling of Poincar\'e disk with $n$-gons ($n \geq 5$), $k \geq 4$ of them meeting in each vertex. Then for any two indices, there exists at most one path connecting them.
\end{lem}
The above lemma is a straightforward generalization of Lemma 1 from \cite{rafal-karol-mykhailo--holocode}. Assume that there are two paths, then one can consider geodesics (lines) tangent to the inner edges of tiles at the path's intersection and beginning contained between them. Since the paths intersect for the second time, the geodesics must intersect as well, which contradicts the hyperbolic geometry.

\begin{thm}
The two-point correlation function between two boundary subsystems is non-trivial only if there exists a path connecting those subsystems. Moreover, in such a case, the two-point correlation function simplifies to
\begin{equation}
\la \phi(x_1) \phi(x_2) \ra= \Tr[ \tilde{\mathcal{O}}\mathcal{P} (v_1 \otimes v_2)],
\end{equation}
where $\mathcal{P}$ is the operator obtained from the contraction of the path tensors with their conjugate on all boundary indices except the distinct ones, $\tilde{\mathcal{O}} = \text{Tr}_{\perp \mathcal{P}}[\mathcal{O}]$ is the bulk operator traced on all bulk indices except the one belonging to the path
while $v_1$, $v_2$ are traceless probing observables on the boundary.
\end{thm}
To prove this theorem, it is sufficient to follow the steps of the proof of Theorem 2 from \cite{rafal-karol-mykhailo--holocode}. While calculating the correlation function, multiple tensors from the network will be reduced with their conjugates.
The boundaries of the non-reduced part of the network are by definition paths, because paths cannot be reduced using isometry conditions. 
However, since there is only one path which connects two boundary indices, the claim follows.

From now on, we consider a localized bulk operator $\mathcal{O}$, i.e. an operator which is supported on one subsystem in
the bulk. Then almost every node from the path can be further simplified by contracting bulk indices, for example
\begin{equation}
\label{eq:red_node}
N_{s_1,s_3}^{s_1', s_3'} = T^{s_0}_{s_1,s_2,s_3,s_4,s_5,s_6} \overline{T}^{s_0}_{s_1',s_2,s_3',s_4,s_5,s_6}~,
\end{equation}
 which is a reduction of $T$ to two subsystems. Such a node can be considered as a matrix from indices $s_1, s_1'$ into $s_3, s_3'$. It is simple to notice that due to the isometric properties of $T$, node $N$'s largest eigenvalue is $1$, and the corresponding (unnormalized) eigenvector is given by $\sum_i |ii\ra$, since application of this vector corresponds to the reduction of $T$ to only one subsystem. 

\begin{figure}[h]
  \centering
      \includegraphics{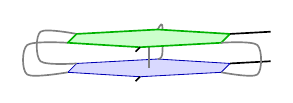}
  \caption{Contraction of the bulk and four pairs of boundary indices between planar hexagonal tensor (violet) and its conjugate (green), which happens for nodes on the path between distinct boundary indices while calculating two or three point correlation functions \eqref{eq:red_node}.}
  \label{fig:hexa2}
\end{figure}

The calculation of correlation reduces to repeated multiplications of the reduced path node and the leading eigenvector gives a zero term since we consider traceless observables.
One may define $\mu$ as a scaling factor of a tensor network network \cite{Central_charges_and_scaling},
i.e. the factor by which the number of boundary indices grow as one adds one layer of a tensors to the network.
Then as one enlarge the network by one layer, the distance between two boundary sides is multiplied by $\mu$ and the path is extended by two nodes.
Thus, we arrive at the following statement

\begin{cor}
\label{cor_2point}
Assume that the path connecting two boundary subsystems is geodesic, and the reduced path node has only one eigenvector to the maximal eigenvalue.
Then, by normalizing the nodes by this factor, in the limit of a large network, one obtains the desired power-law decay of the two-point correlation function,
\begin{equation}
\la\phi(x_1)\phi(x_2)\ra \propto \frac{1}{\ell^{2 \Delta}},
\end{equation}
where $\ell$ is the distance between points, and $\Delta$ is the scaling dimension given by
\begin{equation}
\Delta = - \text{ln}(|\lambda_2|)/\text{ln}(\mu).
\end{equation}
Here $\lambda_2$ denotes the subleading eigenvalue of the simplified node, as in eq. \eqref{eq:red_node}, and $\mu$ is a scaling factor of a network.
\end{cor}

We stress that, in general, the nodes reduced in different ways may have different subleading eigenvalues. Thus, if the
path connecting boundary indices takes turns $1\to3$, the scaling dimension may be different from the cases when  the
path takes turns $1\to 4$. 
Moreover, multiple paths do not follow any geodesic present in the tiling and take arbitrary turns at each tile. Since
for each two boundary indices there is at most one path connecting them, one cannot ``replace'' those irregular paths
with more regular ones. In such cases, if no regularity is present, although the desired decay is still approximately present, one has to explicitly calculate the correlations.
This extra flexibility of our model, which might give various decays along different paths, can be interpreted as an
artifact of discretisation, which differentiates between geodesics. 
A similar feature was also discovered in a different approach that aims to create the CFT directly on the Poincar\'e
disk tiling: quasiperiodic CFT (qCFT) \cite{Jahn2022,Jahn2022_1} (see also \cite{quasi1,quasi2}). In this work, the correlations in a
qCFT possess a fractal self-similarity, which enables recovery of the decay properties of CFT correlation functions only
after averaging over sufficiently large regions. 

It turns out that we can generalize the results for three-point correlation functions as well. 

\begin{thm}
Consider a regular tiling of Poincar\'e disk with $n$-gons $(n \geq 5)$, $k \ge 4$ of them meeting in each vertex.
The three-point correlation function between three distinct boundary indices is nontrivial (nonzero) only if there
exists a path between two of them and a path leading from the third index to some tensor on the first path.  
While calculating the correlation, all tensors will reduce, except the ones on those paths.

Furthermore, there exists no configuration of three boundary indices such that each pair of them is connected by a path.
\end{thm}

Let us start by showing by contradiction that there cannot exist $3$ paths connecting each pair of distinct indices. Lets consider three such paths $t_{1,2}, t_{2,3}, t_{1,3}$, and mark their last  common nodes  $n_1, n_2, n_3$. 

We start by inscribing two geodesics $l_{1,3}, l_{2,3}$ between paths $t_{1,3}$, $t_{2,3}$, such that they contain the inner edges of the first nodes after $n_3$ and are between those paths. 
In a similar manner, we can also replace the path $t_{1,2}$ with a geodesic $l_{1,2}$ tangent to it from the "inner" side of $n_3$ on any chosen tile. Since the paths $t_{1,3}$ and $t_{2,3}$ intersect with $t_{1,2}$ and all the geodesics are tangent to the inner edges of the paths, they must intersect with each other as well.

The angle between geodesics $l_{1,3}$ and $l_{2,3}$ is the same as between paths $t_{1,3}$, $t_{2,3}$ and is given by $2 \pi (1 - 3/k)$. The angles between $l_{1,2}$ and other geodesics
cannot be smaller than $2\pi/k$, which is the minimal angle in the tiling. Thus, the inner angles of a triangle constructed from geodesics are $(1 - 1/k) 2 \pi > \pi$, which contradicts hyperbolic geometry.

The important conclusion from the above property is that there exist no nodes surrounded by paths that do not belong to any of those.
Therefore, one can reduce all tensors except ones on the path connecting two distinct indices, and the path adjoining the last index.

With that, we are ready to present the final statement from this Section:

\begin{cor}
\label{cor_3point}
With the same assumptions as in Corollary \ref{cor_2point}, the obtained three-point correlation function has the desired form,
\begin{equation}
\label{3point_corr_behaviour}
 \la \phi(x_1)\phi(x_2)\phi(x_3) \ra = \frac{C_{123} }{\left(\ell_{12}\,\, \ell_{23} \,\, \ell_{13}\right)^\Delta},
\end{equation}
with the same scaling dimension as for 2-point correlation function, $\Delta = -\text{ln}(|\lambda_2|)/\text{ln}(\mu)$. The proportionality coefficient $C_{123}$ depends
on the reduction of the bulk operator to the node in which paths intersect.
\end{cor}

\subsection{Results}

Presented theory can be directly applied to the concrete tensors derived in previous Sections.
Let us start with planar hexagonal tensors discussed in Section \ref{sec:hexagonal}, which will be placed on $\{6,4\}$ tiling.

In the case of Type I tensors, it is relatively simple to compute the eigenvalues of reduced nodes. In the case of the bipartition $(1,3) \rightarrow (0,2,4,5,6)$, the ratio of the two largest eigenvalues is
\begin{equation}
\lambda^{hexa}_{2} = \left(\sqrt{2}-32 a\right) \sqrt{\left(\sqrt{2}-16 a\right) a} ,
\end{equation}
with $a \in [0, \sqrt{2}/16]$,
whereas for bipartition $(1,4) \rightarrow (0,2,3,5,6)$, the ratio is constant and equal ${\lambda_{2}^{hexb}}   = 1/4$.  These two ratios correspond to scaling dimension equal
\begin{equation}
  \begin{split} 
\Delta_{hexa} &= - \frac{\text{ln}\left[ \left(\sqrt{2}-32 a\right) \sqrt{\left(\sqrt{2}-16 a\right) a}\right]}{\text{ln}(3 + 2 \sqrt{2})} ,\\
         \Delta_{hexb} &= \frac{\text{ln}(4)}{\text{ln}(3 + 2 \sqrt{2})} \approx 0.78~.
 \end{split}    
\end{equation}
The minimum of $\Delta_{hexa}$ is obtained at $a_{min} = \left(\sqrt{2}-1\right)/32$
and is equal $\Delta_{hexa}|_{a = a_{min}} = \text{ln}(8)/\text{ln}(3 + 2 \sqrt{2}) \approx 1.18$, whereas in the limits $a\to 0$ and $a\to \sqrt{2}/16$ it tends to infinity.
Unfortunately, these two values do not coincide for any $a$ within the considered range, so it is troublesome to give them a physical interpretation. Similar phenomena were observed in alternative approaches, e.g. \cite{Jahn2019_2}, and usually are attributed to a fractal structure of the network. The standard workaround to obtain the physical value of the scaling dimension is to average over regions of the boundary instead of considering single indices.

In principle, one may try to reconcile these two values by applying local rotations to the considered tensor. Since the
reduced node is a reshuffling of the reduced density matrix $\rho_{1,3}$, such a rotation may modify its
spectrum. However, it turns out that such operation changes only \textit{third} eigenvalue of node reduced according to
bipartition $(1,4)\to(2,3,5,6)$, and does not stretch the range of second eigenvalue of a node reduced according to
bipartition $(1,3)\to(2,3,5,6)$ enough to encompass $\lambda_{2}^{hexb} = 1/4$.

When considering the tensors corresponding to isolated points on Fig. \ref{fig:entr}, we meet the same difficulty.

{The situation is physically more appealing with the planar pentagonal tensor discussed in Section \ref{sec:penta}.}
  Since it does
not have a ``bulk'' index, it cannot be considered as a part of tensor network mapping bulk Hilbert space to boundary
Hilbert space, but rather as a tensor network defining one specific state on the boundary. That being said, all
arguments from the above discussion still hold, so we can calculate the scaling dimension of correlations for such a
state. 

Due to symmetry of the tensor, we only need to consider one type of reduced node, corresponding to the bipartition $(1,3)\to(2,4,5)$, for which the ratio between the leading and subleading eigenvalues is given by
\begin{equation}
  \lambda^{penta}_{2}  = \frac{ \sqrt{5}-1}{4}~,
\end{equation}
hence the corresponding scaling factor read
\begin{equation}
\label{eq:d_penta}
  \Delta_{penta} = \frac{\text{ln}\left(4/\left(\sqrt{5}-1\right)\right)}{\text{ln} \left(\sqrt{3}+2\right)} \approx 0.89
\end{equation}

Finally, we may combine the obtained pentagonal tensor with some six-qubit perfect tensor (with 5 boundary and one bulk
index) into one node of the tensor network, as in \cite{rafal-karol-mykhailo--holocode}. The perfect tensor provides an
isometry from the bulk into boundary subsystems, whereas the pentagonal tensor ``spoils'' the perfectness of the
combined tensor, resulting in nontrivial correlations.   
Note that the choice of perfect tensor is irrelevant in the
calculation of decay of the two-point correlation functions. 
The last step is to join the pentagonal tensor with the perfect tensor by entangling  each boundary index of a perfect tensors with a boundary index from planar pentagonal tensor, using arbitrary unitary matrix. The violin plots of scaling dimension
obtained from such combined tensors are presented in Fig. \ref{fig:distribution_delta_random}. 

\begin{figure}[h]
\centering
          \includegraphics[height=2.2in]{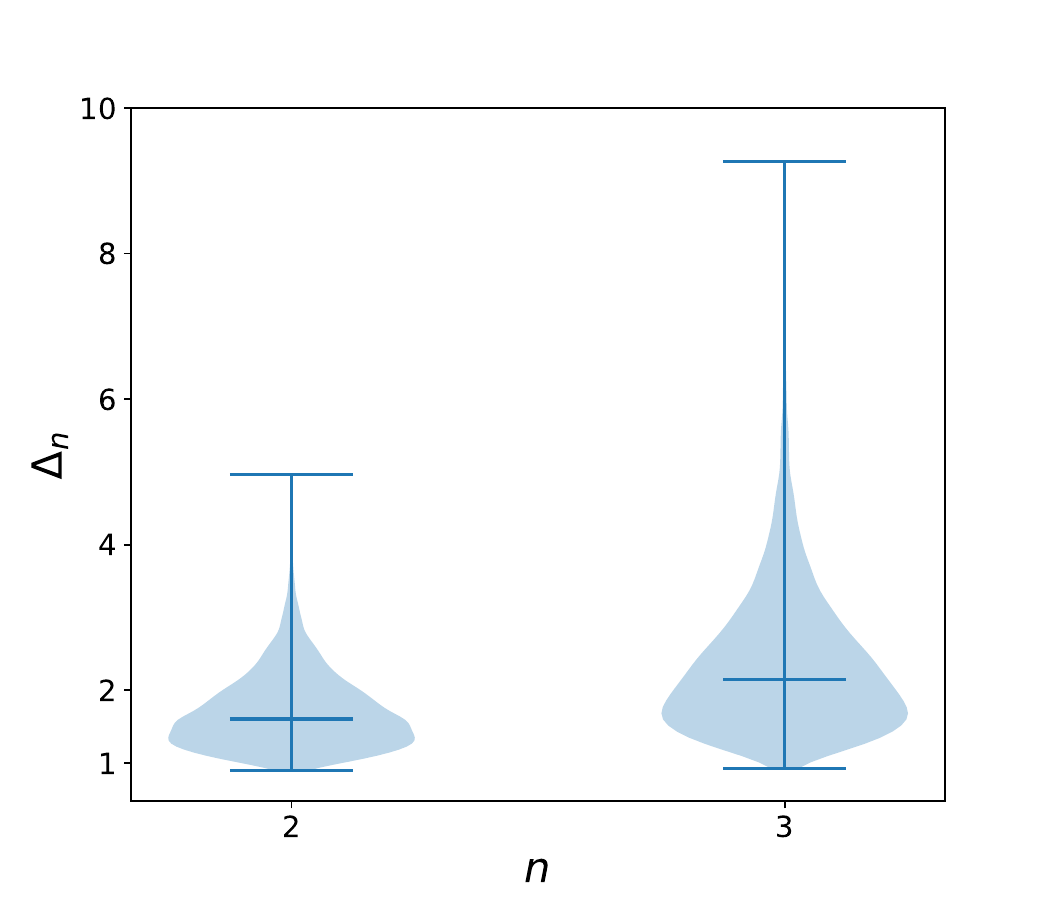}

 \caption{Violin plot of dimensions $\Delta_n$ corresponding to subleasing, $n = 2$, and next to subleasing, $n = 3$, eigenvalue  for nodes constructed as in \eqref{eq:red_node}, where the tensor $T$ was obtained by connecting planar pentagonal tensor with perfect tensor, using random unitary matrices drawn from Haar measure. The number of random samples was $10^6$.
 The maximal, minimal and median values are highlighted, with minimal values lower-bounded by \eqref{eq:d_penta}.
    }
    \label{fig:distribution_delta_random}
\end{figure}

We finish this Section by mentioning that in the Appendix \ref{sec:n-gon} we optimize and extend the construction
from \cite{rafal-karol-mykhailo--holocode}, thus provide general families of parameterized tensors suitable for $n$-gon tiling.
{However, this construction is slightly convoluted, with each boundary subspace consists of at least $3$ qubits leading to minimal local dimension $d = 8$,
so it is described in Appendix.
}

\section{Concluding remarks}

We introduced a framework to describe classes of tensors, corresponding to $n$-partite pure states with special
entanglement properties.
A given class is defined by an $n$-vertex graph with each edge representing certain constraints for the tensors. For
instance, a complete graph of $n$ vertices corresponds to an AME state of $n$ parties with arbitrary local dimension $D$.

In this work we proposed such a general construction of graph restricted tensors and analyzed in detail two specific
examples: 
imperfect pentagonal and hexagonal tensors for $D=2$. By
imposing symmetry requirements on the tensors, it was possible to obtain relatively small sets of equations for the
independent components of the tensor.

In these two concrete examples it was possible to find exact solutions for the explicit form of the tensors
corresponding to pure states of 5 and 7 qubits.
They take the form of one-parameter
families of tensors, 
and isolated points. However, this list is not known to be complete.
The results we obtained allowed us to establish analytically the correlation functions and scaling exponents for the
corresponding tensor network model of the 
holographic principle.
    {Furthermore, we also propose a family of tensors with larger local dimensions satisfying above-mentioned requirements.}

It is known that for the 7-qubit system there are no AME states \cite{AME-7qbits-nogo}.
Therefore, there is a huge interest in finding strongly entangled
states in this system with some extremal properties \cite{borras2007, AMEcomb1, tapiador2009, zha2013,Zhi:2022anp, TSFMMPPZ_26}.
Hence, the present work provides a contribution to this subject,
as we identified 7-qubit states determined by tensors \eqref{aT}, \eqref{ak}, 
for which 6 partial traces over 4 selected subsystems are maximally mixed
and it enjoys a desired rotational symmetry.

Extending the list of analytically solvable models requires further work. Moreover, it would be interesting to find
analytical bounds for the minimal values of the scaling exponents. This issue will be relevant in future searches for 
holographic models applicable to low dimensional conformal field theories.

\begin{acknowledgments}
We would like to thank Alexander Jahn and Arul Lakshminarayan for helpful remarks and fruitful correspondence.
RB acknowledges support by the National Science Center, Poland, under the contract number 2023/50/E/ST2/00472.
BP and MM were supported by the Hungarian National Research,
Development and Innovation Office, NKFIH Grant No. K-145904 and BP was also supported by the NKFIH excellence grant
TKP2021\_NKTA\_64. K{\.Z} acknowledges funding by the European Union under ERC
Advanced Grant Tatypic, Project No. 101142236.  
\end{acknowledgments}

\appendix

\section{Hypergraph-constrained tensors}
\label{app:hyper}

The notion of graph-constrained tensors as formulated in Definition \ref{def:graph_res}  misses certain
interesting tensors. A simple nontrivial example is presented in \cite{ternary-unitary} (see also \cite{sajat-multi}) as
a building block for 2D quantum circuits, such that the arising tensor network has cubic geometry. In this case the
tensor has 8 indices, which are associated with 8 vertices of a cube. The requirement is that we have maximal entanglement
for three different cuts that run parallel to the faces. 

If we try to fit this example into our framework, then a problem arises: the inclusion of all edges constituting
the desired cliques,i,e.  all cube edges and all face-diagonals, results in the creation of additional undesired cliques.
 An example is shown in Figure
\ref{fig:cube_graphs}.

To counter the problem, we present a more  general approach
based on hypergraphs, introduced in some analogy to the
notion of quantum hypergraph states \cite{Rossi13}.

\begin{defn}
An undirected hypergraph is a pair $G = (V,E)$, where $V = \{v_i\}_{i = 1}^n$ is a set of elements called vertices, and $E$ is a set of  subsets of vertices $e_{k_1,\cdots,k_m} := \{v_{{k_1}},\cdots ,v_{{k_m}}\}$ called hyperedges.
\end{defn}

Thus, we can define hypergraph-constrained tensors.

\begin{defn}\label{def:hyper_graph_res}
  Let $G$ be a hypergraph of $n$ vertices. A rank-$n$ tensor $T_{i_1,\cdots,i_n}$ is constrained by a hypergraph $G$ if for any bipartition of its indices $(l_1,...,l_{m'}) \rightarrow (k_{1},...,k_{m})$ with $m \le n/2$ for which $e_{{k_1},\cdots,{k_m}} = \{v_{{k_1}},\cdots,v_{{k_m}}\}$ is a hyperedge of $G$, the corresponding reduction of $T$ is proportional to identity, i.e.
\begin{equation}
\rho_{i_{k_1},...,i_{k_m}}^{j_{k_1},...,j_{k_m}} \propto \delta_{i_{k_1}}^{j_{k_1}}  \cdots \delta_{i_{k_m}}^{j_{k_m}}\,. 
\label{eq_start_v2}
\end{equation}
Furthermore, if for any bipartition $(l_1,...,l_{m'}) \rightarrow (k_{1},...,k_{m})$ satisfying \eqref{eq_start_v2}, the corresponding set of vertices $\{v_{k_1},...,v_{k_m}\}$ is a hyperedge of graph $G$, we say that $T$ is faithfully constrained by graph $G$.
\end{defn}

\begin{figure}
         \centering
         \includegraphics[width=3cm]{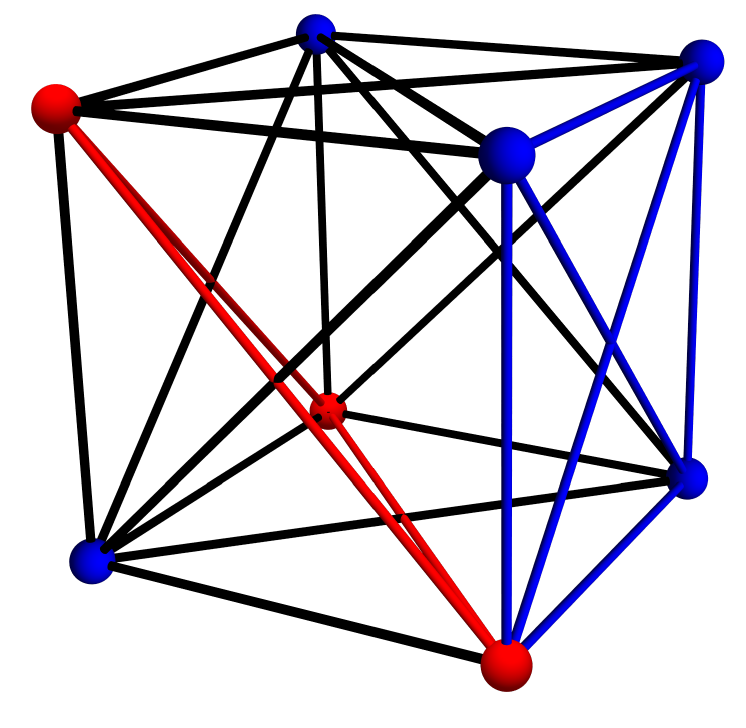}
	 \caption{
           Failed attempt to construct a graph that would encode the entanglement pattern of the tensors of
           \cite{ternary-unitary}. 
 After the introduction of cliques on all faces of the cube (e.g. the blue
           edges), new undesired cliques arise, one of them denoted in red. }
        \label{fig:cube_graphs}
\end{figure}

Comparing this definition with Definition \ref{def:graph_res}, one may notice that the only substantial difference is the exchange of the graph's cliques into hypergraph edges.
Definition \ref{def:hyper_graph_res} is convenient because of its generality, which is expressed in the following observation.

\begin{obs}
  For any tensor $T$ there exists a hypergraph $G$ for which $T$ is faithfully constrained by $G$.
\end{obs}

We can argue for the above observation by a simple construction.
Let us take an arbitrary tensor $T$ on order $n$ and a set $V$ of $n$ vertices. We define the hyperedges as the subsets of indices, for which equation \eqref{eq_start_v2} is satisfied.
By the construction, there are no "additional" hyper-edges corresponding to undesired bipartitions, and representation is faithful.

Finally we may provide the ``composition rules'' for contraction of tensors.

\begin{prop}
Let $T^{(1)}$ and $T^{(2)}$ be two tensors  constrained by hyper-graphs $G^{(1)}$, $G^{(2)}$.
If $T^{(1)}$ and $T^{(2)}$ are contracted on some indices corresponding to one hyper-edge in each graph, $e^{(1)}, e^{(2)}$,
then the resulting tensor is constrained by a hyper-graph $G$ which is a disjoint union of $G^{(1)}$ and $G^{(2)}$
both with removed vertices corresponding to contracted indices.

Furthermore, for every pair of super-edges \newline$\tilde{e}^{(1)}\supset e^{(1)}$, and $\tilde{e}^{(2)}\supset e^{(2)}$
one can connect all non-remover vertices of $\tilde{e}^{(1)}$ and all not remover vertices of $\tilde{e}^{(2)}$ creating a new hyper-edge $\tilde{e}$,  strengthening constraints on resulting tensor.
\end{prop}

The proof of this proposition follows from the same arguments as for graph-constrained tensors and corresponding graphs.

\section{Solutions}
\label{app:solutions}

Here we list further solutions of the constraints imposed on hexagonal tensors. For the tensor components we use the notations
introduced in \eqref{aT}.

\subsection{Type II}

We have observed numerically that these solutions are characterized by the relation
\begin{equation}
  (a_1 - a_2)^{2} + 3 (a_3 - a_4)^{2} = \frac 18 \,.
\end{equation}
These solutions have $\Delta s_{013}\neq 0$ and $\Delta s_{014} \neq 0$ and thus the hyperbolic tiling with the corresponding tensors would yield non-trivial correlation functions.
However, we were not able to find any analytical formula for these solutions.

\subsection{Type III}

These solutions satisfy additional relations
\begin{equation}
  \begin{aligned}
    |a_{1}| = |a_{10}|\,,\quad |a_{2}| = |a_{9}|\,,\quad |a_{3}|=|a_{4}|=|a_{12}|
    \label{eq:typeiii}
  \end{aligned}
\end{equation}
For these solutions,the entropy of reduced state $\rho_{013}$ is equal $\Delta s_{013}=0$. Therefore, the corresponding tensors are not interesting as building blocks of holographic codes. There are several solutions satisfying conditions \eqref{eq:typeiii}. For instance, a one-parameter solution of this type reads
\begin{equation}
  \begin{aligned}
    a_1 &= a_{10} =  a
    \,, \quad a_{2} = a_9 = a + \frac{\sqrt{2}}{8}\,,  \\ a_{3}&= \frac{\sqrt{3 \!-\!16\sqrt{2} a \!-\! 128 a^2}}{16 \sqrt{2}}\,, \\
    a_4 &= a_{11} = -a_{12} = a_{13 } = a_{3}\,,\\
    a_{5} &= - a_6 = a_7 = -a_8 = - \frac{\sqrt{2}}{16}\,,
  \end{aligned}
\end{equation}
where $-\frac{3 \sqrt{2}}{16} < a < \frac{\sqrt{2}}{16}$.

\subsection{Isolated points}

We found two solutions for graph-restricted $7$-index tensor, which in Fig. \ref{fig:entr} corresponds to isolated
points marked as P1 and P2.  
For Point P2 we obtained an exact solution, which satisfies the constraints \eqref{eq:typei}.
 Below we provide explicit formulas for the case $j=1, k=0$. One of them  (P1A) reads, in variables of \eqref{aT}
\begin{equation}
  \begin{aligned}
    a_1&= \frac{\sqrt{445 \sqrt{2}+650}+10}{160 \sqrt{2}} \,, \\ 
a_3&= -\frac{\sqrt{50-5 \sqrt{2}}+20}{160 \sqrt{2}} \,, \\ 
a_5&= \frac{\sqrt{5 \left(50-31 \sqrt{2}\right)}-10}{160 \sqrt{2}} \,, \\ 
a_7&=\frac{1}{32} \left(\sqrt{ \frac{10 + \sqrt{2}}{10}} - \sqrt{2}\right)\,, \\
a_9&= -\frac{\sqrt{5 \left(\sqrt{2}+2\right)}-6}{32 \sqrt{2}} \,, \\ 
a_{10}&= -\frac{\sqrt{5 \left(\sqrt{2}+2\right)}+6}{32 \sqrt{2}} \,, \\ 
a_{11}&= \frac{1}{32} \sqrt{5+\frac{31}{5 \sqrt{2}}} \,, \\ 
a_{12}&= -\frac{1}{32} \sqrt{5-\frac{5}{\sqrt{2}}}\,, \\
a_{13}&= -\frac{1}{32} \sqrt{13+\frac{79}{5 \sqrt{2}}} \,.
  \end{aligned}
\end{equation}
The other analytical solution (P2B) we have found in the case $j=k=1$ is given by
\begin{equation}
  \begin{aligned}
   a_1 &=  -\frac{\sqrt{890-205 \sqrt{2}}-10}{160 \sqrt{2}} \,, \\ 
a_3 &=  -\frac{\sqrt{85 \sqrt{2}+130}+20}{160 \sqrt{2}} \,, \\ 
a_5 &=  \frac{\sqrt{10-5 \sqrt{2}}-10}{160 \sqrt{2}} \,, \\ 
    a_7 &=  \frac{1}{32} \left( \sqrt{\frac{10-\sqrt{2}}{10}} - \sqrt{2}\right) \,, \\
    a_9 &=  -\frac{\sqrt{115 \sqrt{2}+650}-30}{160 \sqrt{2}} \,, \\
    a_{10} &=  -\frac{\sqrt{115 \sqrt{2}+650}+30}{160 \sqrt{2}} \,, \\ 
  \end{aligned}
\end{equation}
\begin{equation*}
  \begin{aligned}
a_{11} &=  -\frac{1}{32} \sqrt{\frac{1}{10} \left(82-31 \sqrt{2}\right)} \,, \\ 
a_{12} &=  \frac{3}{32} \sqrt{\frac{1}{10} \left(\sqrt{2}+2\right)} \,, \\ 
a_{13} &=  \frac{1}{32} \sqrt{13-\frac{79}{5 \sqrt{2}}}~,
  \end{aligned}
\end{equation*}
for which the values of entropies for reduced density matrices are  $\Delta s_{013} = \frac{53}{6400} = 0.00828125$ and $\Delta s_{014} =\frac{19}{160} = 0.11875$ (see Figure \ref{fig:entr}).

\section{Construction of n-gon tiling}\label{sec:n-gon}

In this Appendix, we present a construction of family of tensors suited for $n$-gon tiling of Poincar\'e disk, with $n\geq 4$, with multiple free parameters, by generalizing the construction from \cite{rafal-karol-mykhailo--holocode}. The presented scheme is based on qubits; however, one can use any other finite-size Hilbert spaces as well.  This approach have natural advantage of strong generality and flexibility. On the other hand, it requires at minimum $3$ qubits for each boundary index, resulting in local dimension $8$.

We start with a qubit graph state (not to be confused with a graph-restricted tensor), which we refer to as $G$, defined
by the graph of the shape of a wheel with spokes,  as in  Fig. \ref{fig:hexa}. Each qubit on the circumference corresponds to a boundary degree of freedom, whereas the central qubit corresponds to the bulk one.
    {If one wishes to have different dimension of bulk subsystem, then this tensor may be replaced by graph state defined by the same graph but with chosen local dimension.}

By simple calculation, one can immediately check that such a state, as a tensor $G$, satisfies the isometry requirements $(0,i,i+1)\to(2,3,\dots)$.
 However, we have to spoil this tensor by supplementing it with another one, since $G$ is also an isometry in various different directions, which may lead to oversimplification of the tensor network while calculating correlation functions.

To do so, following \cite{rafal-karol-mykhailo--holocode}, we define a "frame" tensor $F$ on ${n}$ \textit{pairs} of qubits. It consists of dual-unitary matrices $U_{i_1 i_2}^{\;i_3 i_4}$, which have unitary property for two different sets of indices
\begin{equation*}
  \begin{split}
    \sum_{j_1,j_2} U_{i_1, i_2}^{j_1,j_2} \overline{U}_{k_1 k_2}^{j_1,j_2} &= \delta_{i_1,k_1}\delta_{i_2,k_2}\,, \\
    \;\sum_{j_1,j_2} U_{i_1, j_1}^{i_2,j_2} \overline{U}_{k_1 j_i}^{k_2,j_2} &= \delta_{i_1,k_1}\delta_{i_2,k_2} \,,
 \end{split}
    \end{equation*}
as presented diagrammatically in Fig.  \ref{fig:dual_uni}.
\begin{figure}[h!]
    \centering
        \centering
        \includegraphics[height=1.2 cm]{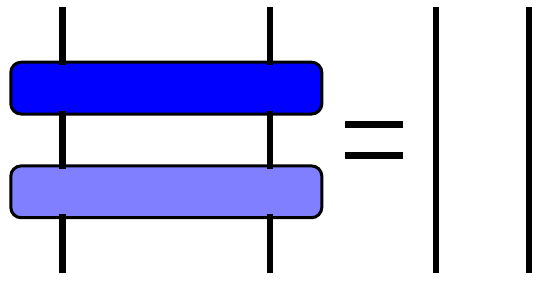}
        \hspace{1 cm}
        \includegraphics[height=1 cm]{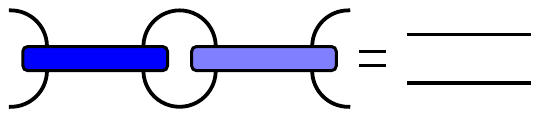}
    \caption{Two isometry properties of dual unitary matrices. In these and subsequent figures, each line represents one index, and the sum is performed over indices which lines do not have loose ends. An exception is a line with two loose ends that corresponds to the Kronecker delta. The pale color represents conjugation.}
    \label{fig:dual_uni}
\end{figure}

We start our construction with $n$ outgoing ququart indices. Each quart is split into two qubits, one connecting it with the neighboring outgoing index on the right and one on the left, as presented in Fig. \ref{fig:frame_node_n}. Thus, we have $n$ connections, each between a pair of adjacent boundary indices of $F$. The next step is to intertwine these connections, using arbitrary dual unitary matrices. For each connection, one by one from the outside, dual unitary matrices $U_i$ ($i = 1,\cdots, n-3$) entangle it with the remaining connections, except neighboring left and right ones, see Fig. \ref{fig:frame_node_n}.
Thus, the tensor $F$ consists of $n(n-3)/2$ dual unitary matrices. To preserve the rotational symmetry of the construction we set the same matrices on all connections, which implies $U_1 = U_{n-3}^\top, U_2 = U_{n-4}^\top, U_3 = U_{n-5}^\top,\cdots$, so we have at most $\lceil (n-3)/2\rceil$ different dual-unitary matrices.

\begin{figure}[h!]
    \centering
        \centering
        \includegraphics[height=5 cm]{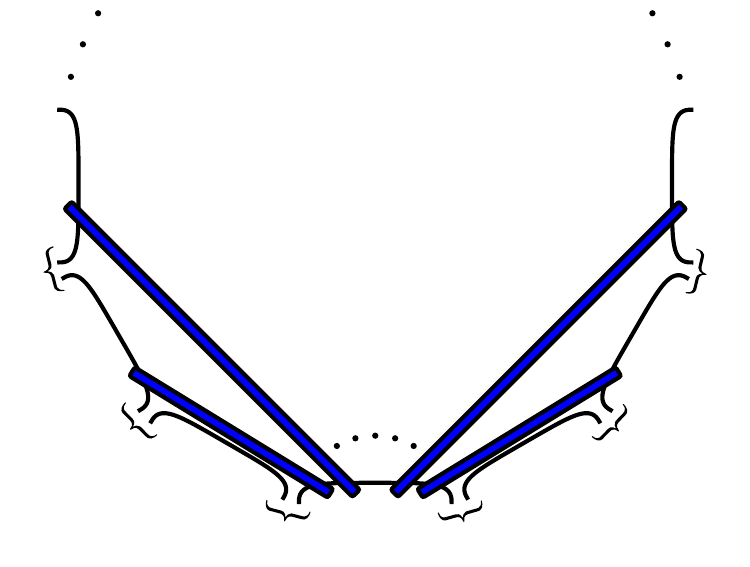}
    \caption{Part of the frame tensor $F$ construction, presenting first and last two dual-unitaries entangling one connection between two neighboring ququart indices with other connections.  The ququarts are decomposed into pairs of qubit, as denoted by the curly brackets, and then qubits from neighboring ququarts are combined by dual unitary-matrices represented as blue boxes.}
    \label{fig:frame_node_n}
\end{figure}

To show that the frame tensor $F$ serves as an isometry $(1,2)\to(3,4,\cdots)$, let us consider the contraction of $F$ with its conjugate on all indices except two neighboring ones. Let's first consider unitaries acting on connections which are not between distinct indices $(1,2)$. Using the first isometric property from Fig \ref{fig:dual_uni}, one can reduce these unitaries, starting with the outer, "shorter" ones simultaneously on all connections except one. After repeating this step, all dual unitaries are reduced with their conjugates, except for those that act on the connection between two non-contracted indices $(1,2)$.
However, since the ends of those unitaries opposite to distinct connection belong to different connections, they can be reduced using the second isometry property from Fig. \ref{fig:dual_uni}. After the reduction of all unitaries, the remaining indices constitute an identity between the uncontracted indices of $F$ and its conjugate.

Furthermore, to show that the frame tensor $F$ is not, in general, an isometry from two non-neighboring indices to the remaining ones, let us present two unitaries, each entangling different connections from distinct indices, which cross, as presented in Fig. \ref{fig:nonreduction_of_frame}. While calculating the contraction of the frame with its conjugate to check the isometry property, even if all other dual-unitaries would reduce with their conjugates, those two cannot, unless they are unitary also after reshuffling. Thus, if the frame is not built from AME states, it is not an isometry in any undesired direction.

\begin{figure}[h!]
    \centering
        \centering
        \includegraphics[height=5 cm]{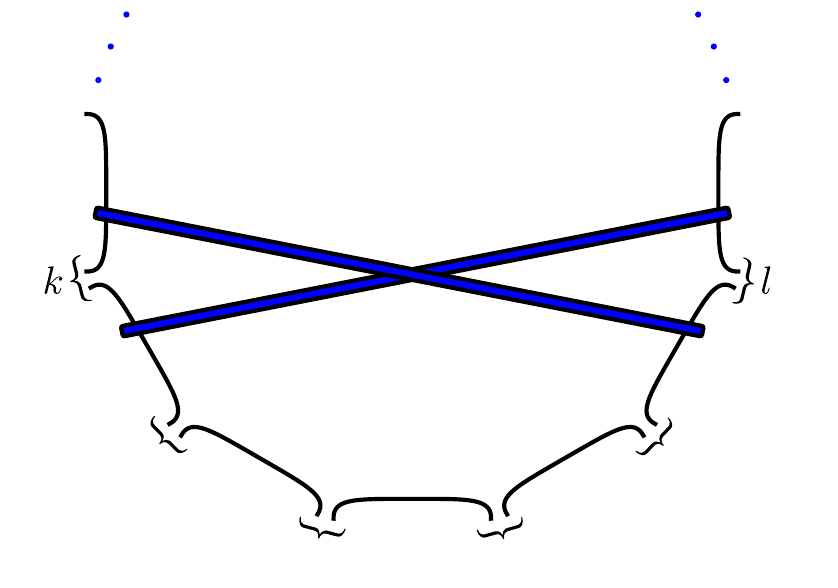}
    \caption{Two exemplary dual-unitaries which cannot be reduced while checking the isometry condition from $k^{\text{th}}$ and $l^{\text{th}}$ ququarts into remaining subsystems. Same as before pairs of qubits forming ququarts are denoted by curly brackets.}
    \label{fig:nonreduction_of_frame}
\end{figure}

Since we constructed an appropriate frame tensor, the last step is to entangle each (ququart) index of the frame $F$ with boundary indices of the graph state $G$, using arbitrary three-qubit unitaries, to create one consistent tensor.


\begin{thebibliography}{46}%
\makeatletter
\providecommand \@ifxundefined [1]{%
 \@ifx{#1\undefined}
}%
\providecommand \@ifnum [1]{%
 \ifnum #1\expandafter \@firstoftwo
 \else \expandafter \@secondoftwo
 \fi
}%
\providecommand \@ifx [1]{%
 \ifx #1\expandafter \@firstoftwo
 \else \expandafter \@secondoftwo
 \fi
}%
\providecommand \natexlab [1]{#1}%
\providecommand \enquote  [1]{``#1''}%
\providecommand \bibnamefont  [1]{#1}%
\providecommand \bibfnamefont [1]{#1}%
\providecommand \citenamefont [1]{#1}%
\providecommand \href@noop [0]{\@secondoftwo}%
\providecommand \href [0]{\begingroup \@sanitize@url \@href}%
\providecommand \@href[1]{\@@startlink{#1}\@@href}%
\providecommand \@@href[1]{\endgroup#1\@@endlink}%
\providecommand \@sanitize@url [0]{\catcode `\\12\catcode `\$12\catcode
  `\&12\catcode `\#12\catcode `\^12\catcode `\_12\catcode `\%12\relax}%
\providecommand \@@startlink[1]{}%
\providecommand \@@endlink[0]{}%
\providecommand \url  [0]{\begingroup\@sanitize@url \@url }%
\providecommand \@url [1]{\endgroup\@href {#1}{\urlprefix }}%
\providecommand \urlprefix  [0]{URL }%
\providecommand \Eprint [0]{\href }%
\providecommand \doibase [0]{https://doi.org/}%
\providecommand \selectlanguage [0]{\@gobble}%
\providecommand \bibinfo  [0]{\@secondoftwo}%
\providecommand \bibfield  [0]{\@secondoftwo}%
\providecommand \translation [1]{[#1]}%
\providecommand \BibitemOpen [0]{}%
\providecommand \bibitemStop [0]{}%
\providecommand \bibitemNoStop [0]{.\EOS\space}%
\providecommand \EOS [0]{\spacefactor3000\relax}%
\providecommand \BibitemShut  [1]{\csname bibitem#1\endcsname}%
\let\auto@bib@innerbib\@empty
\bibitem [{\citenamefont {Cirac}\ \emph {et~al.}(2021)\citenamefont {Cirac},
  \citenamefont {P\'erez-Garc\'{\i}a}, \citenamefont {Schuch},\ and\
  \citenamefont {Verstraete}}]{mps-intro-2020}%
  \BibitemOpen
  \bibfield  {author} {\bibinfo {author} {\bibfnamefont {J.~I.}\ \bibnamefont
  {Cirac}}, \bibinfo {author} {\bibfnamefont {D.}~\bibnamefont
  {P\'erez-Garc\'{\i}a}}, \bibinfo {author} {\bibfnamefont {N.}~\bibnamefont
  {Schuch}},\ and\ \bibinfo {author} {\bibfnamefont {F.}~\bibnamefont
  {Verstraete}},\ }\bibfield  {title} {\bibinfo {title} {Matrix product states
  and projected entangled pair states: Concepts, symmetries, theorems},\ }\href
  {https://doi.org/10.1103/RevModPhys.93.045003} {\bibfield  {journal}
  {\bibinfo  {journal} {Rev. Mod. Phys.}\ }\textbf {\bibinfo {volume} {93}},\
  \bibinfo {pages} {045003} (\bibinfo {year} {2021})},\ \Eprint
  {https://arxiv.org/abs/2011.12127} {arXiv:2011.12127 [quant-ph]} \BibitemShut
  {NoStop}%
\bibitem [{\citenamefont {{Or{\'u}s}}(2019)}]{tensor-network-review}%
  \BibitemOpen
  \bibfield  {author} {\bibinfo {author} {\bibfnamefont {R.}~\bibnamefont
  {{Or{\'u}s}}},\ }\bibfield  {title} {\bibinfo {title} {{Tensor networks for
  complex quantum systems}},\ }\href
  {https://doi.org/10.1038/s42254-019-0086-7} {\bibfield  {journal} {\bibinfo
  {journal} {Nat. Rev. Phys.}\ }\textbf {\bibinfo {volume} {1}},\ \bibinfo
  {pages} {538} (\bibinfo {year} {2019})},\ \Eprint
  {https://arxiv.org/abs/1812.04011} {arXiv:1812.04011 [cond-mat.str-el]}
  \BibitemShut {NoStop}%
\bibitem [{\citenamefont {{Paeckel}}\ \emph {et~al.}(2019)\citenamefont
  {{Paeckel}}, \citenamefont {{K{\"o}hler}}, \citenamefont {{Swoboda}},
  \citenamefont {{Manmana}}, \citenamefont {{Schollw{\"o}ck}},\ and\
  \citenamefont {{Hubig}}}]{mps-time-evolution-review}%
  \BibitemOpen
  \bibfield  {author} {\bibinfo {author} {\bibfnamefont {S.}~\bibnamefont
  {{Paeckel}}}, \bibinfo {author} {\bibfnamefont {T.}~\bibnamefont
  {{K{\"o}hler}}}, \bibinfo {author} {\bibfnamefont {A.}~\bibnamefont
  {{Swoboda}}}, \bibinfo {author} {\bibfnamefont {S.~R.}\ \bibnamefont
  {{Manmana}}}, \bibinfo {author} {\bibfnamefont {U.}~\bibnamefont
  {{Schollw{\"o}ck}}},\ and\ \bibinfo {author} {\bibfnamefont {C.}~\bibnamefont
  {{Hubig}}},\ }\bibfield  {title} {\bibinfo {title} {{Time-evolution methods
  for matrix-product states}},\ }\href
  {https://doi.org/10.1016/j.aop.2019.167998} {\bibfield  {journal} {\bibinfo
  {journal} {Ann. Phys.}\ }\textbf {\bibinfo {volume} {411}},\ \bibinfo {eid}
  {167998} (\bibinfo {year} {2019})},\ \Eprint
  {https://arxiv.org/abs/1901.05824} {arXiv:1901.05824 [cond-mat.str-el]}
  \BibitemShut {NoStop}%
\bibitem [{\citenamefont {{Berezutskii}}\ \emph {et~al.}(2025)\citenamefont
  {{Berezutskii}}, \citenamefont {{Liu}}, \citenamefont {{Acharya}},
  \citenamefont {{Ellerbrock}}, \citenamefont {{Gray}}, \citenamefont
  {{Haghshenas}}, \citenamefont {{He}}, \citenamefont {{Khan}}, \citenamefont
  {{Kuzmin}}, \citenamefont {{Lyakh}}, \citenamefont {{Lykov}}, \citenamefont
  {{Mandr{\`a}}}, \citenamefont {{Mansell}}, \citenamefont {{Melnikov}},
  \citenamefont {{Melnikov}}, \citenamefont {{Mironov}}, \citenamefont
  {{Morozov}}, \citenamefont {{Neukart}}, \citenamefont {{Nocera}},
  \citenamefont {{Perlin}}, \citenamefont {{Perelshtein}}, \citenamefont
  {{Steinberg}}, \citenamefont {{Shaydulin}}, \citenamefont {{Villalonga}},
  \citenamefont {{Pflitsch}}, \citenamefont {{Pistoia}}, \citenamefont
  {{Vinokur}},\ and\ \citenamefont {{Alexeev}}}]{tensor-network-qcomp}%
  \BibitemOpen
  \bibfield  {author} {\bibinfo {author} {\bibfnamefont {A.}~\bibnamefont
  {{Berezutskii}}}, \bibinfo {author} {\bibfnamefont {M.}~\bibnamefont
  {{Liu}}}, \bibinfo {author} {\bibfnamefont {A.}~\bibnamefont {{Acharya}}},
  \bibinfo {author} {\bibfnamefont {R.}~\bibnamefont {{Ellerbrock}}}, \bibinfo
  {author} {\bibfnamefont {J.}~\bibnamefont {{Gray}}}, \bibinfo {author}
  {\bibfnamefont {R.}~\bibnamefont {{Haghshenas}}}, \bibinfo {author}
  {\bibfnamefont {Z.}~\bibnamefont {{He}}}, \bibinfo {author} {\bibfnamefont
  {A.}~\bibnamefont {{Khan}}}, \bibinfo {author} {\bibfnamefont
  {V.}~\bibnamefont {{Kuzmin}}}, \bibinfo {author} {\bibfnamefont
  {D.}~\bibnamefont {{Lyakh}}}, \bibinfo {author} {\bibfnamefont
  {D.}~\bibnamefont {{Lykov}}}, \bibinfo {author} {\bibfnamefont
  {S.}~\bibnamefont {{Mandr{\`a}}}}, \bibinfo {author} {\bibfnamefont
  {C.}~\bibnamefont {{Mansell}}}, \bibinfo {author} {\bibfnamefont
  {A.}~\bibnamefont {{Melnikov}}}, \bibinfo {author} {\bibfnamefont
  {A.}~\bibnamefont {{Melnikov}}}, \bibinfo {author} {\bibfnamefont
  {V.}~\bibnamefont {{Mironov}}}, \bibinfo {author} {\bibfnamefont
  {D.}~\bibnamefont {{Morozov}}}, \bibinfo {author} {\bibfnamefont
  {F.}~\bibnamefont {{Neukart}}}, \bibinfo {author} {\bibfnamefont
  {A.}~\bibnamefont {{Nocera}}}, \bibinfo {author} {\bibfnamefont {M.~A.}\
  \bibnamefont {{Perlin}}}, \bibinfo {author} {\bibfnamefont {M.}~\bibnamefont
  {{Perelshtein}}}, \bibinfo {author} {\bibfnamefont {M.}~\bibnamefont
  {{Steinberg}}}, \bibinfo {author} {\bibfnamefont {R.}~\bibnamefont
  {{Shaydulin}}}, \bibinfo {author} {\bibfnamefont {B.}~\bibnamefont
  {{Villalonga}}}, \bibinfo {author} {\bibfnamefont {M.}~\bibnamefont
  {{Pflitsch}}}, \bibinfo {author} {\bibfnamefont {M.}~\bibnamefont
  {{Pistoia}}}, \bibinfo {author} {\bibfnamefont {V.}~\bibnamefont
  {{Vinokur}}},\ and\ \bibinfo {author} {\bibfnamefont {Y.}~\bibnamefont
  {{Alexeev}}},\ }\bibfield  {title} {\bibinfo {title} {{Tensor networks for
  quantum computing}},\ }\href {https://doi.org/10.1038/s42254-025-00853-1}
  {\bibfield  {journal} {\bibinfo  {journal} {Nat. Rev. Phys.}\ }\textbf
  {\bibinfo {volume} {7}},\ \bibinfo {pages} {581–593} (\bibinfo {year}
  {2025})},\ \Eprint {https://arxiv.org/abs/2503.08626} {arXiv:2503.08626
  [quant-ph]} \BibitemShut {NoStop}%
\bibitem [{\citenamefont {{Jahn}}\ and\ \citenamefont
  {{Eisert}}(2021)}]{holocode-review-jahn-eisert}%
  \BibitemOpen
  \bibfield  {author} {\bibinfo {author} {\bibfnamefont {A.}~\bibnamefont
  {{Jahn}}}\ and\ \bibinfo {author} {\bibfnamefont {J.}~\bibnamefont
  {{Eisert}}},\ }\bibfield  {title} {\bibinfo {title} {{Holographic tensor
  network models and quantum error correction: a topical review}},\ }\href
  {https://doi.org/10.1088/2058-9565/ac0293} {\bibfield  {journal} {\bibinfo
  {journal} {Quantum Sci. Tech.}\ }\textbf {\bibinfo {volume} {6}},\ \bibinfo
  {eid} {033002} (\bibinfo {year} {2021})},\ \Eprint
  {https://arxiv.org/abs/2102.02619} {arXiv:2102.02619 [quant-ph]} \BibitemShut
  {NoStop}%
\bibitem [{\citenamefont {{Pastawski}}\ \emph {et~al.}(2015)\citenamefont
  {{Pastawski}}, \citenamefont {{Yoshida}}, \citenamefont {{Harlow}},\ and\
  \citenamefont {{Preskill}}}]{ads-code-1}%
  \BibitemOpen
  \bibfield  {author} {\bibinfo {author} {\bibfnamefont {F.}~\bibnamefont
  {{Pastawski}}}, \bibinfo {author} {\bibfnamefont {B.}~\bibnamefont
  {{Yoshida}}}, \bibinfo {author} {\bibfnamefont {D.}~\bibnamefont
  {{Harlow}}},\ and\ \bibinfo {author} {\bibfnamefont {J.}~\bibnamefont
  {{Preskill}}},\ }\bibfield  {title} {\bibinfo {title} {{Holographic quantum
  error-correcting codes: toy models for the bulk/boundary correspondence}},\
  }\href {https://doi.org/10.1007/JHEP06(2015)149} {\bibfield  {journal}
  {\bibinfo  {journal} {JHEP}\ }\textbf {\bibinfo {volume} {2015}},\ \bibinfo
  {eid} {149}},\ \Eprint {https://arxiv.org/abs/1503.06237} {arXiv:1503.06237
  [hep-th]} \BibitemShut {NoStop}%
\bibitem [{\citenamefont {Maldacena}(1999)}]{Maldacena-ads-cft}%
  \BibitemOpen
  \bibfield  {author} {\bibinfo {author} {\bibfnamefont {J.}~\bibnamefont
  {Maldacena}},\ }\bibfield  {title} {\bibinfo {title} {The large n limit of
  superconformal field theories and supergravity},\ }\href
  {https://doi.org/10.1023/a:1026654312961} {\bibfield  {journal} {\bibinfo
  {journal} {International Journal of Theoretical Physics}\ }\textbf {\bibinfo
  {volume} {38}},\ \bibinfo {pages} {1113–1133} (\bibinfo {year} {1999})},\
  \Eprint {https://arxiv.org/abs/hep-th/9711200} {arXiv:hep-th/9711200
  [hep-th]} \BibitemShut {NoStop}%
\bibitem [{\citenamefont {Witten}(1998)}]{Witten1998}%
  \BibitemOpen
  \bibfield  {author} {\bibinfo {author} {\bibfnamefont {E.}~\bibnamefont
  {Witten}},\ }\bibfield  {title} {\bibinfo {title} {Anti de {S}itter space and
  holography},\ }\href {https://doi.org/10.4310/ATMP.1998.v2.n2.a2} {\bibfield
  {journal} {\bibinfo  {journal} {ATMP}\ }\textbf {\bibinfo {volume} {2}},\
  \bibinfo {pages} {253} (\bibinfo {year} {1998})}\BibitemShut {NoStop}%
\bibitem [{\citenamefont {{Bertini}}\ \emph {et~al.}(2025)\citenamefont
  {{Bertini}}, \citenamefont {{Claeys}},\ and\ \citenamefont
  {{Prosen}}}]{dual-unitary-review}%
  \BibitemOpen
  \bibfield  {author} {\bibinfo {author} {\bibfnamefont {B.}~\bibnamefont
  {{Bertini}}}, \bibinfo {author} {\bibfnamefont {P.~W.}\ \bibnamefont
  {{Claeys}}},\ and\ \bibinfo {author} {\bibfnamefont {T.}~\bibnamefont
  {{Prosen}}},\ }\bibfield  {title} {\bibinfo {title} {{Exactly solvable
  many-body dynamics from space-time duality}},\ }\bibfield  {journal}
  {\bibinfo  {journal} {arXiv e-prints}\ }\href
  {https://doi.org/10.48550/arXiv.2505.11489} {10.48550/arXiv.2505.11489}
  (\bibinfo {year} {2025}),\ \Eprint {https://arxiv.org/abs/2505.11489}
  {arXiv:2505.11489 [cond-mat.stat-mech]} \BibitemShut {NoStop}%
\bibitem [{\citenamefont {{Yu}}\ \emph {et~al.}(2024)\citenamefont {{Yu}},
  \citenamefont {{Cirac}}, \citenamefont {{Kos}},\ and\ \citenamefont
  {{Styliaris}}}]{du-peps}%
  \BibitemOpen
  \bibfield  {author} {\bibinfo {author} {\bibfnamefont {X.-H.}\ \bibnamefont
  {{Yu}}}, \bibinfo {author} {\bibfnamefont {J.~I.}\ \bibnamefont {{Cirac}}},
  \bibinfo {author} {\bibfnamefont {P.}~\bibnamefont {{Kos}}},\ and\ \bibinfo
  {author} {\bibfnamefont {G.}~\bibnamefont {{Styliaris}}},\ }\bibfield
  {title} {\bibinfo {title} {{Dual-Isometric Projected Entangled Pair
  States}},\ }\href {https://doi.org/10.1103/PhysRevLett.133.190401} {\bibfield
   {journal} {\bibinfo  {journal} {Phys. Rev. Lett.}\ }\textbf {\bibinfo
  {volume} {133}},\ \bibinfo {eid} {190401} (\bibinfo {year} {2024})},\ \Eprint
  {https://arxiv.org/abs/2404.16783} {arXiv:2404.16783 [quant-ph]} \BibitemShut
  {NoStop}%
\bibitem [{\citenamefont {{Mesty{\'a}n}}\ \emph {et~al.}(2024)\citenamefont
  {{Mesty{\'a}n}}, \citenamefont {{Pozsgay}},\ and\ \citenamefont
  {{Wanless}}}]{sajat-multi}%
  \BibitemOpen
  \bibfield  {author} {\bibinfo {author} {\bibfnamefont {M.}~\bibnamefont
  {{Mesty{\'a}n}}}, \bibinfo {author} {\bibfnamefont {B.}~\bibnamefont
  {{Pozsgay}}},\ and\ \bibinfo {author} {\bibfnamefont {I.~M.}\ \bibnamefont
  {{Wanless}}},\ }\bibfield  {title} {\bibinfo {title} {{Multi-directional
  unitarity and maximal entanglement in spatially symmetric quantum states}},\
  }\href {https://doi.org/10.21468/SciPostPhys.16.1.010} {\bibfield  {journal}
  {\bibinfo  {journal} {SciPost Phys.}\ }\textbf {\bibinfo {volume} {16}},\
  \bibinfo {eid} {010} (\bibinfo {year} {2024})},\ \Eprint
  {https://arxiv.org/abs/2210.13017} {arXiv:2210.13017 [quant-ph]} \BibitemShut
  {NoStop}%
\bibitem [{\citenamefont {{Rajchel-Mieldzio{\'c}}}\ \emph
  {et~al.}(2025)\citenamefont {{Rajchel-Mieldzio{\'c}}}, \citenamefont
  {{Bistro{\'n}}}, \citenamefont {{Rico}}, \citenamefont {{Lakshminarayan}},\
  and\ \citenamefont {{{\.Z}yczkowski}}}]{ame-review}%
  \BibitemOpen
  \bibfield  {author} {\bibinfo {author} {\bibfnamefont {G.}~\bibnamefont
  {{Rajchel-Mieldzio{\'c}}}}, \bibinfo {author} {\bibfnamefont
  {R.}~\bibnamefont {{Bistro{\'n}}}}, \bibinfo {author} {\bibfnamefont
  {A.}~\bibnamefont {{Rico}}}, \bibinfo {author} {\bibfnamefont
  {A.}~\bibnamefont {{Lakshminarayan}}},\ and\ \bibinfo {author} {\bibfnamefont
  {K.}~\bibnamefont {{{\.Z}yczkowski}}},\ }\bibfield  {title} {\bibinfo {title}
  {{Absolutely maximally entangled pure states of multipartite quantum
  systems}},\ }\bibfield  {journal} {\bibinfo  {journal} {arXiv e-prints}\
  }\href {https://doi.org/10.48550/arXiv.2508.04777}
  {10.48550/arXiv.2508.04777} (\bibinfo {year} {2025}),\ \Eprint
  {https://arxiv.org/abs/2508.04777} {arXiv:2508.04777 [quant-ph]} \BibitemShut
  {NoStop}%
\bibitem [{\citenamefont {{Aravinda}}\ \emph {et~al.}(2021)\citenamefont
  {{Aravinda}}, \citenamefont {{Rather}},\ and\ \citenamefont
  {{Lakshminarayan}}}]{dual-unitary--bernoulli}%
  \BibitemOpen
  \bibfield  {author} {\bibinfo {author} {\bibfnamefont {S.}~\bibnamefont
  {{Aravinda}}}, \bibinfo {author} {\bibfnamefont {S.~A.}\ \bibnamefont
  {{Rather}}},\ and\ \bibinfo {author} {\bibfnamefont {A.}~\bibnamefont
  {{Lakshminarayan}}},\ }\bibfield  {title} {\bibinfo {title} {{From
  dual-unitary to quantum Bernoulli circuits: Role of the entangling power in
  constructing a quantum ergodic hierarchy}},\ }\href
  {https://doi.org/10.1103/PhysRevResearch.3.043034} {\bibfield  {journal}
  {\bibinfo  {journal} {Phys. Rev. Res.}\ }\textbf {\bibinfo {volume} {3}},\
  \bibinfo {eid} {043034} (\bibinfo {year} {2021})},\ \Eprint
  {https://arxiv.org/abs/2101.04580} {arXiv:2101.04580 [quant-ph]} \BibitemShut
  {NoStop}%
\bibitem [{\citenamefont {Cao}\ \emph {et~al.}(2022)\citenamefont {Cao},
  \citenamefont {Pollack},\ and\ \citenamefont {Wang}}]{ChunJun2022}%
  \BibitemOpen
  \bibfield  {author} {\bibinfo {author} {\bibfnamefont {C.}~\bibnamefont
  {Cao}}, \bibinfo {author} {\bibfnamefont {J.}~\bibnamefont {Pollack}},\ and\
  \bibinfo {author} {\bibfnamefont {Y.}~\bibnamefont {Wang}},\ }\bibfield
  {title} {\bibinfo {title} {Hyperinvariant multiscale entanglement
  renormalization ansatz: Approximate holographic error correction codes with
  power-law correlations},\ }\href
  {https://doi.org/10.1103/physrevd.105.026018} {\bibfield  {journal} {\bibinfo
   {journal} {Phys. Rev. D}\ }\textbf {\bibinfo {volume} {105}},\ \bibinfo
  {pages} {026018} (\bibinfo {year} {2022})}\BibitemShut {NoStop}%
\bibitem [{\citenamefont {{Berger}}\ and\ \citenamefont
  {{Osborne}}(2018)}]{perfect-tangles}%
  \BibitemOpen
  \bibfield  {author} {\bibinfo {author} {\bibfnamefont {J.}~\bibnamefont
  {{Berger}}}\ and\ \bibinfo {author} {\bibfnamefont {T.~J.}\ \bibnamefont
  {{Osborne}}},\ }\bibfield  {title} {\bibinfo {title} {{Perfect tangles}},\
  }\href@noop {} {\bibfield  {journal} {\bibinfo  {journal} {arXiv e-prints}\ }
  (\bibinfo {year} {2018})},\ \Eprint {https://arxiv.org/abs/1804.03199}
  {arXiv:1804.03199 [quant-ph]} \BibitemShut {NoStop}%
\bibitem [{\citenamefont {Harris}\ \emph {et~al.}(2018)\citenamefont {Harris},
  \citenamefont {McMahon}, \citenamefont {Brennen},\ and\ \citenamefont
  {Stace}}]{block-perfect-tensor}%
  \BibitemOpen
  \bibfield  {author} {\bibinfo {author} {\bibfnamefont {R.~J.}\ \bibnamefont
  {Harris}}, \bibinfo {author} {\bibfnamefont {N.~A.}\ \bibnamefont {McMahon}},
  \bibinfo {author} {\bibfnamefont {G.~K.}\ \bibnamefont {Brennen}},\ and\
  \bibinfo {author} {\bibfnamefont {T.~M.}\ \bibnamefont {Stace}},\ }\bibfield
  {title} {\bibinfo {title} {Calderbank-{S}hor-{S}teane holographic quantum
  error-correcting codes},\ }\href {https://doi.org/10.1103/PhysRevA.98.052301}
  {\bibfield  {journal} {\bibinfo  {journal} {Phys. Rev. A}\ }\textbf {\bibinfo
  {volume} {98}},\ \bibinfo {pages} {052301} (\bibinfo {year} {2018})},\
  \Eprint {https://arxiv.org/abs/1806.06472} {arXiv:1806.06472 [quant-ph]}
  \BibitemShut {NoStop}%
\bibitem [{\citenamefont {{Wang}}(2021)}]{planar-OA}%
  \BibitemOpen
  \bibfield  {author} {\bibinfo {author} {\bibfnamefont {Y.-L.}\ \bibnamefont
  {{Wang}}},\ }\bibfield  {title} {\bibinfo {title} {{Planar k-uniform states:
  a generalization of planar maximally entangled states}},\ }\href
  {https://doi.org/10.1007/s11128-021-03204-y} {\bibfield  {journal} {\bibinfo
  {journal} {Quant. Inf. Proc.}\ }\textbf {\bibinfo {volume} {20}},\ \bibinfo
  {eid} {271} (\bibinfo {year} {2021})},\ \Eprint
  {https://arxiv.org/abs/2106.12209} {arXiv:2106.12209 [quant-ph]} \BibitemShut
  {NoStop}%
\bibitem [{\citenamefont {{Doroudiani}}\ and\ \citenamefont
  {{Karimipour}}(2020)}]{planar-AME}%
  \BibitemOpen
  \bibfield  {author} {\bibinfo {author} {\bibfnamefont {M.}~\bibnamefont
  {{Doroudiani}}}\ and\ \bibinfo {author} {\bibfnamefont {V.}~\bibnamefont
  {{Karimipour}}},\ }\bibfield  {title} {\bibinfo {title} {{Planar maximally
  entangled states}},\ }\href {https://doi.org/10.1103/PhysRevA.102.012427}
  {\bibfield  {journal} {\bibinfo  {journal} {Phys. Rev. A}\ }\textbf {\bibinfo
  {volume} {102}},\ \bibinfo {eid} {012427} (\bibinfo {year} {2020})},\ \Eprint
  {https://arxiv.org/abs/2004.00906} {arXiv:2004.00906 [quant-ph]} \BibitemShut
  {NoStop}%
\bibitem [{\citenamefont {{Evenbly}}(2017)}]{Evenbly-HyMera}%
  \BibitemOpen
  \bibfield  {author} {\bibinfo {author} {\bibfnamefont {G.}~\bibnamefont
  {{Evenbly}}},\ }\bibfield  {title} {\bibinfo {title} {{Hyperinvariant Tensor
  Networks and Holography}},\ }\href
  {https://doi.org/10.1103/PhysRevLett.119.141602} {\bibfield  {journal}
  {\bibinfo  {journal} {Phys. Rev. Lett.}\ }\textbf {\bibinfo {volume} {119}},\
  \bibinfo {eid} {141602} (\bibinfo {year} {2017})},\ \Eprint
  {https://arxiv.org/abs/1704.04229} {arXiv:1704.04229 [quant-ph]} \BibitemShut
  {NoStop}%
\bibitem [{\citenamefont {Hayden}\ \emph {et~al.}(2016)\citenamefont {Hayden},
  \citenamefont {Nezami}, \citenamefont {Qi}, \citenamefont {Thomas},
  \citenamefont {Walter},\ and\ \citenamefont {Yang}}]{Hayden:2016cfa}%
  \BibitemOpen
  \bibfield  {author} {\bibinfo {author} {\bibfnamefont {P.}~\bibnamefont
  {Hayden}}, \bibinfo {author} {\bibfnamefont {S.}~\bibnamefont {Nezami}},
  \bibinfo {author} {\bibfnamefont {X.-L.}\ \bibnamefont {Qi}}, \bibinfo
  {author} {\bibfnamefont {N.}~\bibnamefont {Thomas}}, \bibinfo {author}
  {\bibfnamefont {M.}~\bibnamefont {Walter}},\ and\ \bibinfo {author}
  {\bibfnamefont {Z.}~\bibnamefont {Yang}},\ }\bibfield  {title} {\bibinfo
  {title} {Holographic duality from random tensor networks},\ }\bibfield
  {journal} {\bibinfo  {journal} {Journal of High Energy Physics}\ }\textbf
  {\bibinfo {volume} {2016}},\ \href {https://doi.org/10.1007/jhep11(2016)009}
  {10.1007/jhep11(2016)009} (\bibinfo {year} {2016})\BibitemShut {NoStop}%
\bibitem [{\citenamefont {Qi}\ and\ \citenamefont {Yang}(2018)}]{Qi:2018shh}%
  \BibitemOpen
  \bibfield  {author} {\bibinfo {author} {\bibfnamefont {X.-L.}\ \bibnamefont
  {Qi}}\ and\ \bibinfo {author} {\bibfnamefont {Z.}~\bibnamefont {Yang}},\
  }\bibfield  {title} {\bibinfo {title} {{Space-time random tensor networks and
  holographic duality}},\ }\href@noop {} {\bibfield  {journal} {\bibinfo
  {journal} {arXiv preprint: 1801.05289}\ } (\bibinfo {year}
  {2018})}\BibitemShut {NoStop}%
\bibitem [{\citenamefont {Bao}\ \emph {et~al.}(2019)\citenamefont {Bao},
  \citenamefont {Penington}, \citenamefont {Sorce},\ and\ \citenamefont
  {Wall}}]{Bao:2018pvs}%
  \BibitemOpen
  \bibfield  {author} {\bibinfo {author} {\bibfnamefont {N.}~\bibnamefont
  {Bao}}, \bibinfo {author} {\bibfnamefont {G.}~\bibnamefont {Penington}},
  \bibinfo {author} {\bibfnamefont {J.}~\bibnamefont {Sorce}},\ and\ \bibinfo
  {author} {\bibfnamefont {A.~C.}\ \bibnamefont {Wall}},\ }\bibfield  {title}
  {\bibinfo {title} {Beyond toy models: distilling tensor networks in full
  {AdS/CFT}},\ }\href {https://doi.org/10.1007/jhep11(2019)069} {\bibfield
  {journal} {\bibinfo  {journal} {Journal of High Energy Physics}\ }\textbf
  {\bibinfo {volume} {11}},\ \bibinfo {pages} {69} (\bibinfo {year}
  {2019})}\BibitemShut {NoStop}%
\bibitem [{\citenamefont {Jahn}\ \emph {et~al.}(2019)\citenamefont {Jahn},
  \citenamefont {Gluza}, \citenamefont {Pastawski},\ and\ \citenamefont
  {Eisert}}]{Jahn2019_2}%
  \BibitemOpen
  \bibfield  {author} {\bibinfo {author} {\bibfnamefont {A.}~\bibnamefont
  {Jahn}}, \bibinfo {author} {\bibfnamefont {M.}~\bibnamefont {Gluza}},
  \bibinfo {author} {\bibfnamefont {F.}~\bibnamefont {Pastawski}},\ and\
  \bibinfo {author} {\bibfnamefont {J.}~\bibnamefont {Eisert}},\ }\bibfield
  {title} {\bibinfo {title} {Holography and criticality in matchgate tensor
  networks},\ }\href {https://doi.org/10.1126/sciadv.aaw0092} {\bibfield
  {journal} {\bibinfo  {journal} {Science Advances}\ }\textbf {\bibinfo
  {volume} {5}},\ \bibinfo {pages} {eaaw0092} (\bibinfo {year}
  {2019})}\BibitemShut {NoStop}%
\bibitem [{\citenamefont {Jahn}\ \emph
  {et~al.}(2022{\natexlab{a}})\citenamefont {Jahn}, \citenamefont
  {Zimbor{\'a}s},\ and\ \citenamefont {Eisert}}]{Jahn2022}%
  \BibitemOpen
  \bibfield  {author} {\bibinfo {author} {\bibfnamefont {A.}~\bibnamefont
  {Jahn}}, \bibinfo {author} {\bibfnamefont {Z.}~\bibnamefont {Zimbor{\'a}s}},\
  and\ \bibinfo {author} {\bibfnamefont {J.}~\bibnamefont {Eisert}},\
  }\bibfield  {title} {\bibinfo {title} {Tensor network models of {AdS/qCFT}},\
  }\href {https://doi.org/10.22331/q-2022-02-03-643} {\bibfield  {journal}
  {\bibinfo  {journal} {Quantum}\ }\textbf {\bibinfo {volume} {6}},\ \bibinfo
  {pages} {643} (\bibinfo {year} {2022}{\natexlab{a}})}\BibitemShut {NoStop}%
\bibitem [{\citenamefont {Jahn}\ \emph
  {et~al.}(2022{\natexlab{b}})\citenamefont {Jahn}, \citenamefont {Gluza},
  \citenamefont {Verhoeven}, \citenamefont {Singh},\ and\ \citenamefont
  {Eisert}}]{Jahn2022_1}%
  \BibitemOpen
  \bibfield  {author} {\bibinfo {author} {\bibfnamefont {A.}~\bibnamefont
  {Jahn}}, \bibinfo {author} {\bibfnamefont {M.}~\bibnamefont {Gluza}},
  \bibinfo {author} {\bibfnamefont {C.}~\bibnamefont {Verhoeven}}, \bibinfo
  {author} {\bibfnamefont {S.}~\bibnamefont {Singh}},\ and\ \bibinfo {author}
  {\bibfnamefont {J.}~\bibnamefont {Eisert}},\ }\bibfield  {title} {\bibinfo
  {title} {Boundary theories of critical matchgate tensor networks},\ }\href
  {https://doi.org/10.1007/jhep04(2022)111} {\bibfield  {journal} {\bibinfo
  {journal} {Journal of High Energy Physics}\ }\textbf {\bibinfo {volume}
  {2022}},\ \bibinfo {pages} {111} (\bibinfo {year} {2022}{\natexlab{b}})},\
  \Eprint {https://arxiv.org/abs/2110.02972} {arXiv:2110.02972 [quant-ph]}
  \BibitemShut {NoStop}%
\bibitem [{\citenamefont {Chen}\ \emph {et~al.}(2024)\citenamefont {Chen},
  \citenamefont {Hung}, \citenamefont {Jiang},\ and\ \citenamefont
  {Lao}}]{Chen:2024unp}%
  \BibitemOpen
  \bibfield  {author} {\bibinfo {author} {\bibfnamefont {L.}~\bibnamefont
  {Chen}}, \bibinfo {author} {\bibfnamefont {L.-Y.}\ \bibnamefont {Hung}},
  \bibinfo {author} {\bibfnamefont {Y.}~\bibnamefont {Jiang}},\ and\ \bibinfo
  {author} {\bibfnamefont {B.-X.}\ \bibnamefont {Lao}},\ }\bibfield  {title}
  {\bibinfo {title} {{Quantum 2D Liouville Path-Integral Is a Sum over
  Geometries in AdS$_3$ Einstein Gravity}},\ }\href@noop {} {\bibfield
  {journal} {\bibinfo  {journal} {arXiv preprint: 2403.03179}\ } (\bibinfo
  {year} {2024})}\BibitemShut {NoStop}%
\bibitem [{\citenamefont {{Bistro{\'n}}}\ \emph {et~al.}(2025)\citenamefont
  {{Bistro{\'n}}}, \citenamefont {{Hontarenko}},\ and\ \citenamefont
  {{{\.Z}yczkowski}}}]{rafal-karol-mykhailo--holocode}%
  \BibitemOpen
  \bibfield  {author} {\bibinfo {author} {\bibfnamefont {R.}~\bibnamefont
  {{Bistro{\'n}}}}, \bibinfo {author} {\bibfnamefont {M.}~\bibnamefont
  {{Hontarenko}}},\ and\ \bibinfo {author} {\bibfnamefont {K.}~\bibnamefont
  {{{\.Z}yczkowski}}},\ }\bibfield  {title} {\bibinfo {title} {Bulk-boundary
  correspondence from hyperinvariant tensor networks},\ }\href
  {https://doi.org/10.1103/PhysRevD.111.026006} {\bibfield  {journal} {\bibinfo
   {journal} {Phys. Rev. D}\ }\textbf {\bibinfo {volume} {111}},\ \bibinfo
  {pages} {026006} (\bibinfo {year} {2025})},\ \Eprint
  {https://arxiv.org/abs/2409.02029} {arXiv:2409.02029 [quant-ph]} \BibitemShut
  {NoStop}%
\bibitem [{\citenamefont {Helwig}(2013)}]{AME-graph}%
  \BibitemOpen
  \bibfield  {author} {\bibinfo {author} {\bibfnamefont {W.}~\bibnamefont
  {Helwig}},\ }\bibfield  {title} {\bibinfo {title} {Absolutely maximally
  entangled qudit graph states},\ }\href@noop {} {\bibfield  {journal}
  {\bibinfo  {journal} {arxiv e-prints}\ } (\bibinfo {year} {2013})},\ \Eprint
  {https://arxiv.org/abs/1306.2879} {arXiv:1306.2879 [quant-ph]} \BibitemShut
  {NoStop}%
\bibitem [{\citenamefont {{Facchi}}\ \emph {et~al.}(2008)\citenamefont
  {{Facchi}}, \citenamefont {{Florio}}, \citenamefont {{Parisi}},\ and\
  \citenamefont {{Pascazio}}}]{AME-1}%
  \BibitemOpen
  \bibfield  {author} {\bibinfo {author} {\bibfnamefont {P.}~\bibnamefont
  {{Facchi}}}, \bibinfo {author} {\bibfnamefont {G.}~\bibnamefont {{Florio}}},
  \bibinfo {author} {\bibfnamefont {G.}~\bibnamefont {{Parisi}}},\ and\
  \bibinfo {author} {\bibfnamefont {S.}~\bibnamefont {{Pascazio}}},\ }\bibfield
   {title} {\bibinfo {title} {{Maximally multipartite entangled states}},\
  }\href {https://doi.org/10.1103/PhysRevA.77.060304} {\bibfield  {journal}
  {\bibinfo  {journal} {Phys. Rev. A}\ }\textbf {\bibinfo {volume} {77}},\
  \bibinfo {eid} {060304} (\bibinfo {year} {2008})},\ \Eprint
  {https://arxiv.org/abs/0710.2868} {arXiv:0710.2868 [quant-ph]} \BibitemShut
  {NoStop}%
\bibitem [{\citenamefont {{Goyeneche}}\ \emph {et~al.}(2015)\citenamefont
  {{Goyeneche}}, \citenamefont {{Alsina}}, \citenamefont {{Latorre}},
  \citenamefont {{Riera}},\ and\ \citenamefont {{{\.Z}yczkowski}}}]{AMEcomb2}%
  \BibitemOpen
  \bibfield  {author} {\bibinfo {author} {\bibfnamefont {D.}~\bibnamefont
  {{Goyeneche}}}, \bibinfo {author} {\bibfnamefont {D.}~\bibnamefont
  {{Alsina}}}, \bibinfo {author} {\bibfnamefont {J.~I.}\ \bibnamefont
  {{Latorre}}}, \bibinfo {author} {\bibfnamefont {A.}~\bibnamefont {{Riera}}},\
  and\ \bibinfo {author} {\bibfnamefont {K.}~\bibnamefont {{{\.Z}yczkowski}}},\
  }\bibfield  {title} {\bibinfo {title} {{Absolutely maximally entangled
  states, combinatorial designs, and multiunitary matrices}},\ }\href
  {https://doi.org/10.1103/PhysRevA.92.032316} {\bibfield  {journal} {\bibinfo
  {journal} {Phys. Rev. A}\ }\textbf {\bibinfo {volume} {92}},\ \bibinfo {eid}
  {032316} (\bibinfo {year} {2015})},\ \Eprint
  {https://arxiv.org/abs/1506.08857} {arXiv:1506.08857 [quant-ph]} \BibitemShut
  {NoStop}%
\bibitem [{\citenamefont {Wang}(2021)}]{planar_k_uniform_states}%
  \BibitemOpen
  \bibfield  {author} {\bibinfo {author} {\bibfnamefont {Y.}~\bibnamefont
  {Wang}},\ }\bibfield  {title} {\bibinfo {title} {Planar k-uniform states: a
  generalization of planar maximally entangled states},\ }\href
  {https://doi.org/10.1007/S11128-021-03204-Y} {\bibfield  {journal} {\bibinfo
  {journal} {Quantum Inf. Process.}\ }\textbf {\bibinfo {volume} {20}},\
  \bibinfo {pages} {1} (\bibinfo {year} {2021})}\BibitemShut {NoStop}%
\bibitem [{\citenamefont {{Wolfram Research{,}
  Inc.}}()}]{MathematicaTensorsPaper}%
  \BibitemOpen
  \bibfield  {author} {\bibinfo {author} {\bibnamefont {{Wolfram Research{,}
  Inc.}}},\ }\href@noop {} {\bibinfo {title} {{Mathematica, {V}ersion
  12.1.1.0}}},\ \bibinfo {howpublished}
  {https://www.wolfram.com/mathematica}\BibitemShut {NoStop}%
\bibitem [{\citenamefont {Ramadas}\ and\ \citenamefont
  {Lakshminarayan}(2025)}]{arul-AME5d}%
  \BibitemOpen
  \bibfield  {author} {\bibinfo {author} {\bibfnamefont {N.}~\bibnamefont
  {Ramadas}}\ and\ \bibinfo {author} {\bibfnamefont {A.}~\bibnamefont
  {Lakshminarayan}},\ }\bibfield  {title} {\bibinfo {title} {Local unitary
  equivalence of absolutely maximally entangled states constructed from
  orthogonal arrays},\ }\href {https://doi.org/10.1088/1751-8121/adbf75}
  {\bibfield  {journal} {\bibinfo  {journal} {J. Phys. A: Math. Theor.}\
  }\textbf {\bibinfo {volume} {58}},\ \bibinfo {pages} {125301} (\bibinfo
  {year} {2025})},\ \Eprint {https://arxiv.org/abs/2411.04096}
  {arXiv:2411.04096 [quant-ph]} \BibitemShut {NoStop}%
\bibitem [{\citenamefont {Tan}(2025)}]{tan-AME52}%
  \BibitemOpen
  \bibfield  {author} {\bibinfo {author} {\bibfnamefont {I.}~\bibnamefont
  {Tan}},\ }\href {https://arxiv.org/abs/2507.02185} {\bibinfo {title}
  {Classification of four-qubit pure codes and five-qubit absolutely maximally
  entangled states}} (\bibinfo {year} {2025}),\ \Eprint
  {https://arxiv.org/abs/2507.02185} {arXiv:2507.02185 [quant-ph]} \BibitemShut
  {NoStop}%
\bibitem [{\citenamefont {{Huber}}\ \emph {et~al.}(2017)\citenamefont
  {{Huber}}, \citenamefont {{G{\"u}hne}},\ and\ \citenamefont
  {{Siewert}}}]{AME-7qbits-nogo}%
  \BibitemOpen
  \bibfield  {author} {\bibinfo {author} {\bibfnamefont {F.}~\bibnamefont
  {{Huber}}}, \bibinfo {author} {\bibfnamefont {O.}~\bibnamefont
  {{G{\"u}hne}}},\ and\ \bibinfo {author} {\bibfnamefont {J.}~\bibnamefont
  {{Siewert}}},\ }\bibfield  {title} {\bibinfo {title} {{Absolutely Maximally
  Entangled States of seven qubits do not exist}},\ }\href
  {https://doi.org/10.1103/PhysRevLett.118.200502} {\bibfield  {journal}
  {\bibinfo  {journal} {Phys. Rev. Lett.}\ }\textbf {\bibinfo {volume} {118}},\
  \bibinfo {eid} {200502} (\bibinfo {year} {2017})},\ \Eprint
  {https://arxiv.org/abs/1608.06228} {arXiv:1608.06228 [quant-ph]} \BibitemShut
  {NoStop}%
\bibitem [{\citenamefont {Jahn}\ \emph {et~al.}(2020)\citenamefont {Jahn},
  \citenamefont {Zimbor\'as},\ and\ \citenamefont
  {Eisert}}]{Central_charges_and_scaling}%
  \BibitemOpen
  \bibfield  {author} {\bibinfo {author} {\bibfnamefont {A.}~\bibnamefont
  {Jahn}}, \bibinfo {author} {\bibfnamefont {Z.}~\bibnamefont {Zimbor\'as}},\
  and\ \bibinfo {author} {\bibfnamefont {J.}~\bibnamefont {Eisert}},\
  }\bibfield  {title} {\bibinfo {title} {{Central charges of aperiodic
  holographic tensor network models}},\ }\href
  {https://doi.org/10.1103/PhysRevA.102.042407} {\bibfield  {journal} {\bibinfo
   {journal} {Phys. Rev. A}\ }\textbf {\bibinfo {volume} {102}},\ \bibinfo
  {pages} {042407} (\bibinfo {year} {2020})},\ \Eprint
  {https://arxiv.org/abs/1911.03485} {arXiv:1911.03485 [hep-th]} \BibitemShut
  {NoStop}%
\bibitem [{\citenamefont {{Boyle}}\ \emph {et~al.}(2020)\citenamefont
  {{Boyle}}, \citenamefont {{Dickens}},\ and\ \citenamefont
  {{Flicker}}}]{quasi1}%
  \BibitemOpen
  \bibfield  {author} {\bibinfo {author} {\bibfnamefont {L.}~\bibnamefont
  {{Boyle}}}, \bibinfo {author} {\bibfnamefont {M.}~\bibnamefont {{Dickens}}},\
  and\ \bibinfo {author} {\bibfnamefont {F.}~\bibnamefont {{Flicker}}},\
  }\bibfield  {title} {\bibinfo {title} {{Conformal Quasicrystals and
  Holography}},\ }\href {https://doi.org/10.1103/PhysRevX.10.011009} {\bibfield
   {journal} {\bibinfo  {journal} {Phys. Rev. X}\ }\textbf {\bibinfo {volume}
  {10}},\ \bibinfo {pages} {011009} (\bibinfo {year} {2020})},\ \Eprint
  {https://arxiv.org/abs/1805.02665} {arXiv:1805.02665 [hep-th]} \BibitemShut
  {NoStop}%
\bibitem [{\citenamefont {{Boyle}}\ and\ \citenamefont
  {{Kulp}}(2025)}]{quasi2}%
  \BibitemOpen
  \bibfield  {author} {\bibinfo {author} {\bibfnamefont {L.}~\bibnamefont
  {{Boyle}}}\ and\ \bibinfo {author} {\bibfnamefont {J.}~\bibnamefont
  {{Kulp}}},\ }\bibfield  {title} {\bibinfo {title} {{Holographic foliations:
  Self-similar quasicrystals from hyperbolic honeycombs}},\ }\href
  {https://doi.org/10.1103/PhysRevD.111.046001} {\bibfield  {journal} {\bibinfo
   {journal} {Phys. Rev. D}\ }\textbf {\bibinfo {volume} {111}},\ \bibinfo
  {eid} {046001} (\bibinfo {year} {2025})},\ \Eprint
  {https://arxiv.org/abs/2408.15316} {arXiv:2408.15316 [hep-th]} \BibitemShut
  {NoStop}%
\bibitem [{\citenamefont {{Borras}}\ \emph {et~al.}(2007)\citenamefont
  {{Borras}}, \citenamefont {{Plastino}}, \citenamefont {{Batle}},
  \citenamefont {{Zander}}, \citenamefont {{Casas}},\ and\ \citenamefont
  {{Plastino}}}]{borras2007}%
  \BibitemOpen
  \bibfield  {author} {\bibinfo {author} {\bibfnamefont {A.}~\bibnamefont
  {{Borras}}}, \bibinfo {author} {\bibfnamefont {A.~R.}\ \bibnamefont
  {{Plastino}}}, \bibinfo {author} {\bibfnamefont {J.}~\bibnamefont {{Batle}}},
  \bibinfo {author} {\bibfnamefont {C.}~\bibnamefont {{Zander}}}, \bibinfo
  {author} {\bibfnamefont {M.}~\bibnamefont {{Casas}}},\ and\ \bibinfo {author}
  {\bibfnamefont {A.}~\bibnamefont {{Plastino}}},\ }\bibfield  {title}
  {\bibinfo {title} {{Multiqubit systems: highly entangled states and
  entanglement distribution}},\ }\href
  {https://doi.org/10.1088/1751-8113/40/44/018} {\bibfield  {journal} {\bibinfo
   {journal} {Journal of Physics A Mathematical General}\ }\textbf {\bibinfo
  {volume} {40}},\ \bibinfo {pages} {13407} (\bibinfo {year} {2007})},\ \Eprint
  {https://arxiv.org/abs/0803.3979} {arXiv:0803.3979 [quant-ph]} \BibitemShut
  {NoStop}%
\bibitem [{\citenamefont {{Goyeneche}}\ and\ \citenamefont
  {{{\.Z}yczkowski}}(2014)}]{AMEcomb1}%
  \BibitemOpen
  \bibfield  {author} {\bibinfo {author} {\bibfnamefont {D.}~\bibnamefont
  {{Goyeneche}}}\ and\ \bibinfo {author} {\bibfnamefont {K.}~\bibnamefont
  {{{\.Z}yczkowski}}},\ }\bibfield  {title} {\bibinfo {title} {{Genuinely
  multipartite entangled states and orthogonal arrays}},\ }\href
  {https://doi.org/10.1103/PhysRevA.90.022316} {\bibfield  {journal} {\bibinfo
  {journal} {Phys. Rev. A}\ }\textbf {\bibinfo {volume} {90}},\ \bibinfo {eid}
  {022316} (\bibinfo {year} {2014})},\ \Eprint
  {https://arxiv.org/abs/1404.3586} {arXiv:1404.3586 [quant-ph]} \BibitemShut
  {NoStop}%
\bibitem [{\citenamefont {{Tapiador}}\ \emph {et~al.}(2009)\citenamefont
  {{Tapiador}}, \citenamefont {{Hernandez-Castro}}, \citenamefont {{Clark}},\
  and\ \citenamefont {{Stepney}}}]{tapiador2009}%
  \BibitemOpen
  \bibfield  {author} {\bibinfo {author} {\bibfnamefont {J.~E.}\ \bibnamefont
  {{Tapiador}}}, \bibinfo {author} {\bibfnamefont {J.~C.}\ \bibnamefont
  {{Hernandez-Castro}}}, \bibinfo {author} {\bibfnamefont {J.~A.}\ \bibnamefont
  {{Clark}}},\ and\ \bibinfo {author} {\bibfnamefont {S.}~\bibnamefont
  {{Stepney}}},\ }\bibfield  {title} {\bibinfo {title} {{Highly entangled
  multi-qubit states with simple algebraic structure}},\ }\href
  {https://doi.org/10.1088/1751-8113/42/41/415301} {\bibfield  {journal}
  {\bibinfo  {journal} {Journal of Physics A Mathematical General}\ }\textbf
  {\bibinfo {volume} {42}},\ \bibinfo {eid} {415301} (\bibinfo {year}
  {2009})},\ \Eprint {https://arxiv.org/abs/0904.3874} {arXiv:0904.3874
  [quant-ph]} \BibitemShut {NoStop}%
\bibitem [{\citenamefont {{Zha}}\ \emph {et~al.}(2013)\citenamefont {{Zha}},
  \citenamefont {{Yuan}},\ and\ \citenamefont {{Zhang}}}]{zha2013}%
  \BibitemOpen
  \bibfield  {author} {\bibinfo {author} {\bibfnamefont {X.}~\bibnamefont
  {{Zha}}}, \bibinfo {author} {\bibfnamefont {C.}~\bibnamefont {{Yuan}}},\ and\
  \bibinfo {author} {\bibfnamefont {Y.}~\bibnamefont {{Zhang}}},\ }\bibfield
  {title} {\bibinfo {title} {{Generalized criterion for a maximally multi-qubit
  entangled state}},\ }\href {https://doi.org/10.1088/1612-2011/10/4/045201}
  {\bibfield  {journal} {\bibinfo  {journal} {Laser Physics Letters}\ }\textbf
  {\bibinfo {volume} {10}},\ \bibinfo {eid} {045201} (\bibinfo {year}
  {2013})},\ \Eprint {https://arxiv.org/abs/1204.6340} {arXiv:1204.6340
  [quant-ph]} \BibitemShut {NoStop}%
\bibitem [{\citenamefont {Zhi}\ and\ \citenamefont {Hu}(2022)}]{Zhi:2022anp}%
  \BibitemOpen
  \bibfield  {author} {\bibinfo {author} {\bibfnamefont {P.}~\bibnamefont
  {Zhi}}\ and\ \bibinfo {author} {\bibfnamefont {Y.}~\bibnamefont {Hu}},\
  }\bibfield  {title} {\bibinfo {title} {{Construct maximally five- and
  seven-qubit entangled states via three- qubit GHZ state}},\ }\href
  {https://doi.org/10.1142/S0217979222502150} {\bibfield  {journal} {\bibinfo
  {journal} {Int. J. Mod. Phys. B}\ }\textbf {\bibinfo {volume} {36}},\
  \bibinfo {pages} {2250215} (\bibinfo {year} {2022})}\BibitemShut {NoStop}%
\bibitem [{\citenamefont {Trotta}\ \emph {et~al.}(2026)\citenamefont {Trotta},
  \citenamefont {Scarafile}, \citenamefont {Facchi}, \citenamefont {Magnifico},
  \citenamefont {Mariano}, \citenamefont {Parisi}, \citenamefont {Pascazio},\
  and\ \citenamefont {Życzkowski}}]{TSFMMPPZ_26}%
  \BibitemOpen
  \bibfield  {author} {\bibinfo {author} {\bibfnamefont {G.}~\bibnamefont
  {Trotta}}, \bibinfo {author} {\bibfnamefont {P.}~\bibnamefont {Scarafile}},
  \bibinfo {author} {\bibfnamefont {P.}~\bibnamefont {Facchi}}, \bibinfo
  {author} {\bibfnamefont {G.}~\bibnamefont {Magnifico}}, \bibinfo {author}
  {\bibfnamefont {A.}~\bibnamefont {Mariano}}, \bibinfo {author} {\bibfnamefont
  {G.}~\bibnamefont {Parisi}}, \bibinfo {author} {\bibfnamefont
  {S.}~\bibnamefont {Pascazio}},\ and\ \bibinfo {author} {\bibfnamefont
  {K.}~\bibnamefont {Życzkowski}},\ }\bibfield  {title} {\bibinfo {title}
  {Multipartite entanglement of random states of qubits},\ }\bibfield
  {journal} {\bibinfo  {journal} {arXiv preprint: 2605.10314}\ }\href
  {https://doi.org/10.48550/ARXIV.2605.10314} {10.48550/ARXIV.2605.10314}
  (\bibinfo {year} {2026})\BibitemShut {NoStop}%
\bibitem [{\citenamefont {{Milbradt}}\ \emph {et~al.}(2023)\citenamefont
  {{Milbradt}}, \citenamefont {{Scheller}}, \citenamefont {{A{\ss}mus}},\ and\
  \citenamefont {{Mendl}}}]{ternary-unitary}%
  \BibitemOpen
  \bibfield  {author} {\bibinfo {author} {\bibfnamefont {R.}~\bibnamefont
  {{Milbradt}}}, \bibinfo {author} {\bibfnamefont {L.}~\bibnamefont
  {{Scheller}}}, \bibinfo {author} {\bibfnamefont {C.}~\bibnamefont
  {{A{\ss}mus}}},\ and\ \bibinfo {author} {\bibfnamefont {C.~B.}\ \bibnamefont
  {{Mendl}}},\ }\bibfield  {title} {\bibinfo {title} {{Ternary unitary quantum
  lattice models and circuits in $2 + 1$ dimensions}},\ }\href
  {https://doi.org/10.1103/PhysRevLett.130.090601} {\bibfield  {journal}
  {\bibinfo  {journal} {Phys. Rev. Lett.}\ }\textbf {\bibinfo {volume} {130}},\
  \bibinfo {pages} {090601} (\bibinfo {year} {2023})},\ \Eprint
  {https://arxiv.org/abs/2206.01499} {arXiv:2206.01499 [cond-mat.stat-mech]}
  \BibitemShut {NoStop}%
\bibitem [{\citenamefont {M~Rossi}\ and\ \citenamefont
  {Macchiavello}(2013)}]{Rossi13}%
  \BibitemOpen
  \bibfield  {author} {\bibinfo {author} {\bibfnamefont {D.~B.}\ \bibnamefont
  {M~Rossi}, \bibfnamefont {M~Huber}}\ and\ \bibinfo {author} {\bibfnamefont
  {C.}~\bibnamefont {Macchiavello}},\ }\bibfield  {title} {\bibinfo {title}
  {Quantum hypergraph states},\ }\href
  {https://doi.org/10.1088/1367-2630/15/11/113022} {\bibfield  {journal}
  {\bibinfo  {journal} {New J. Phys.}\ }\textbf {\bibinfo {volume} {15}},\
  \bibinfo {pages} {113022} (\bibinfo {year} {2013})},\ \Eprint
  {https://arxiv.org/abs/1211.5554} {arXiv:1211.5554 [quant-ph]} \BibitemShut
  {NoStop}%
\end{thebibliography}

%

\end{document}